\documentclass[12pt]{article}
\usepackage{amsmath,amstext,amsthm,amssymb,bbm}
\usepackage{graphicx,subfig}
\usepackage{mathtools}
\usepackage[colorlinks=true,citecolor=blue]{hyperref}
\usepackage[square,numbers,compress]{natbib}
\usepackage{color}
\usepackage{comment,sidecap}
\usepackage{esint}

\PassOptionsToPackage{normalem}{ulem}
\usepackage{ulem}

\usepackage[title]{appendix}
\usepackage{framed}
\usepackage{bookmark}
\usepackage{enumerate}
\usepackage[margin=1.5cm]{geometry}
\usepackage{parskip}

\newcommand{\Lapl}{\mathcal{L}}
\newcommand{\ben}{\begin{enumerate}}
\newcommand{\een}{\end{enumerate}}
\newcommand{\bea}{\begin{eqnarray}}
\newcommand{\eea}{\end{eqnarray}}
\newcommand{\bit}{\begin{itemize}}
\newcommand{\eit}{\end{itemize}}
\newcommand{\be}{\begin{equation}}
\newcommand{\ee}{\end{equation}}
\newcommand{\benn}{\begin{equation*}}
\newcommand{\eenn}{\end{equation*}}
\newcommand{\bdm}{\begin{displaymath}}
\newcommand{\edm}{\end{displaymath}}

\DeclareMathOperator\erfc{erfc}

\DeclarePairedDelimiter\abs{\lvert}{\rvert}
\DeclarePairedDelimiter\norm{\lVert}{\rVert}
\newcommand{\vect}{\mathbf}

\theoremstyle{definition}

\newtheorem{theorem}{Theorem}
\newtheorem{prop}{Proposition}
\newtheorem{lemma}{Lemma}
\newtheorem*{remark}{Remark}
\newtheorem{cor}{Corollary}
\newtheorem{definition}{Definition}

\newcommand{\vc}{\mathbf}

\newcommand{\bxi}{\pmb\xi}
\newcommand{\betta}{\pmb\eta}
\newcommand{\id}{\mbox{d}}

\newcommand{\bnabla}{\pmb\nabla}

\newcommand{\indf}{\mathbbm{1}}
\newcommand{\X}{X_K^{\delta,h}}

\begin{document}
\title{Asymptotic dynamics of inertial particles with memory}
\author{Gabriel Provencher Langlois$^1$,\\ Mohammad Farazmand$^{1,2}\footnote{Corresponding author's email address: farazmam@ethz.ch}$, George Haller$^2$\\
{\small $^1$Department of Mathematics, ETH Zurich, R\"amistrasse 1, 8092 Zurich, Switzerland}\\
{\small $^2$Institute for Mechanical Systems, ETH Zurich, Tannenstrasse 3, 8092 Zurich, Switzerland}
}
\maketitle
\abstract{
Recent experimental and numerical observations have shown the significance of the Basset--Boussinesq memory term on the dynamics of small spherical rigid particles (or inertial particles) suspended in an ambient fluid flow. These observations suggest an algebraic decay to an asymptotic state, as opposed to the exponential convergence in the absence of the memory term. Here, we prove that the observed algebraic decay is a universal property of the Maxey--Riley equation. Specifically, the particle velocity decays algebraically in time to a limit that is $\mathcal O(\epsilon)$-close to the fluid velocity, where $0<\epsilon\ll 1$ is proportional to the square of the ratio of the particle radius to the fluid characteristic length-scale. These results follows from a sharp analytic upper bound that we derive for the particle velocity. For completeness, we also present a first proof of existence and uniqueness of global solutions to the Maxey--Riley equation, a nonlinear system of fractional-order differential equations.
}

\section{Introduction}
\label{sec:introduction}
The motion of a solid body transported by an ambient Newtonian fluid flow can, in principle, be determined by solving the Navier--Stokes equations with appropriate moving boundary conditions \cite{galdi2008hemodynamical,cartwright2010}. The resulting partial differential equations are, however, too complicated for mathematical analysis. Their numerical solutions are computationally expensive and yield little insight.

For the motion of a small spherical rigid body (or inertial particle), however, one can derive a reliable model by accounting for all the forces exerted on the particle due to the solid-fluid interaction. \citet{stokes1851} made the first attempt to obtain such a model for the oscillatory motion of an inertial particle. Later, \citet{basset2}, \citet{boussinesq} and \citet{oseen1927} studied the settling of a solid sphere under gravity in a quiescent fluid. The resulting equation is known as the BBO equation. To study the motion of inertial particles in non-uniform unsteady flow, \citet{Tchen} wrote the BBO equation in a frame of reference moving with the fluid, accounting for various inertial forces that arise in this frame.

The exact form of the forces exerted on the particle has been debated and corrected by several authors (see, e.g., \citet{corrsin1956}). A widely accepted form of the forces was derived by \citet{MR} from first principles. The resulting equation, with the later correction of \citet{auton1988} to the added mass term, is usually referred to as the Maxey--Riley (MR) equation.  

To describe the MR equation, let $\vect{u} : \mathcal{D} \times \mathbb{R}^+ \rightarrow \mathbb{R}^n$ denote a known velocity field describing the flow of a fluid in an open spatial domain $\mathcal{D} \subseteq \mathbb{R}^n$. Here, $n = 2$ or $n = 3$ for two- and three-dimensional flows, respectively. A fluid trajectory is then the solution of the differential equation $\vect{\dot{x}} = \vect{u}(\vect{x}, t)$ with some initial condition $\vect{x}(t_0) = \vect{x_0}$. An inertial particle, however, follows a different trajectory $\vect{y}(t) \in \mathcal{D}$. The particle velocity $\vc v(t)=\dot{\vc y}(t)$ satisfies the Maxey--Riley equation
\be
\begin{split}
\rho_p\dot{\vect{v}} =& \rho_f\frac{\mbox D\vect u}{\mbox Dt}\\
                       & +(\rho_p-\rho_f)\mbox{\textbf g}\\
                       & -\frac{9\nu\rho_f}{2a^2}\left(\vect v-\vect u-\frac{a^2}{6}\Delta\vect u \right)\\
                       & -\frac{\rho_f}{2}\left[\dot{\vect{v}}-\frac{\mbox D\;}{\mbox Dt}\left(\vect u+\frac{a^2}{10}\Delta\vect u\right)\right]\\
                       & -\frac{9\rho_f}{2a}\sqrt{\frac{\nu}{\pi}}\left[\int_{t_0}^t\frac{\dot{\vect w}(s)}{\sqrt{t-s}}\id s+\frac{\vect w(t_0)}{\sqrt{t-t_0}} \right],
\end{split}
\label{eq:MR_original}
\ee
where
\benn
\label{eq:MR_w}
\vect{w}\left(t\right) = \vect{v}\left(t\right) - \vect{u}\left(\vect{y}(t),t \right) - \frac{a^2}{6} \Delta \vect{u}(\vect y(t),t).
\eenn
Here, $\rho_p$ and $\rho_f$ are, respectively, the particle and fluid densities; $\nu$ is the kinematic viscosity of the fluid; $a$ is the particle radius and $\mbox{\textbf g}$ is the constant gravitational acceleration vector. The initial conditions for the inertial particle are given as $\vect{y}(t_0) = \vect{y}_0$ and $\vc v(t_0)=\vect{v}_0$, for some $t_0 \in \mathbb{R}^+$. The material derivative $\frac{\mbox D\;}{\mbox Dt} := \partial_t + \vect{u}\cdot \nabla$ denotes the time derivative along a fluid trajectory.

The right-hand side in \eqref{eq:MR_original} contains the various forces exerted on the particle. The terms written on separate lines are the force exerted by the undisturbed flow on the particle; the buoyancy force; the Stokes drag; the added mass term and the Basset--Boussinesq memory term.

These forces have varying orders of magnitude. In particular, the Basset--Bousinesq memory term, accounting for the lagging boundary layer developed around the sphere, is routinely neglected on the grounds that it is insignificant compared to the Stokes drag and added mass \cite[see, e.g.,][]{IP_maxey87,balkovsky2001}. Recent experimental and numerical studies, however, point to the contrary \cite{IP_candelier,toegel2006,gabrin2009,daitche2011memory,guseva2013influence,Daitche_NJP}.

The numerical simulations of \cite{daitche2011memory,guseva2013influence}, in particular, show the position of the particle to converge to its asymptotic limit algebraically. This is fundamentally different from the exponential convergence arising in the absence of the memory term \cite{rubin1995_IP,mograbi2006,IP_haller08}. In the present paper, we prove that the observations of \cite{daitche2011memory,guseva2013influence} are a universal and generic property of the MR equation with memory, irrespective of the fluid flow carrying the particles. 

The MR equation was originally derived under the assumption $\vc w(t_0)=0$. Later, \citet{MR_initCond} modified the original formulation to lift this unphysical restriction, obtaining equation \eqref{eq:MR_original} above. This equation can be written as a system of nonlinear fractional-order differential equations \cite{kobayashi,MR_EUR} in terms of the particle position $\vc y$ and relative velocity $\vc w$ (see equation \eqref{eq:MR_system_1comp} below). While there exist fundamental results for special classes of fractional-order differential equations (see, e.g., \cite{podlubny1998fractional}), the MR equation does not fit in any of these classes and requires separate treatment.

Even the existence and uniqueness of solutions to the MR equation is unclear. Only recently have \citet{MR_EUR} proved the existence, uniqueness and regularity of its \emph{local} solutions in a weak sense. They also showed that only under the unphysical assumption $\vc w(t_0)=0$ does the MR equation admit strong solutions.

Here, we prove \emph{global} existence and uniqueness of weak solutions to the MR equation. We also prove that the velocity $\vc v$ of a small particle of radius $a$ decays algebraically to an asymptotic state that is $\mathcal O(\frac{a^2}{L^2})$-close to the fluid velocity $\vc u$, where $L$ is a characteristic length scale of the fluid flow.

\section{Preliminaries}
\label{sec:preliminaries}
\subsection{The MR equation in dimensionless variables}
We rewrite the Maxey--Riley equation \eqref{eq:MR_original} in a form more appropriate for mathematical analysis. First, we rescale space, velocities and time using the characteristic length scale $L$, the characteristic velocity $U$ and the characteristic time scale $T=L/U$. Using the resulting dimensionless variables $\mathbf y\mapsto\mathbf y/L$, $\vc u\mapsto \vc u/U$, $\vc v\mapsto\vc v/U$ and $t\mapsto t/T$ and rearranging various terms, we write \eqref{eq:MR_original} as a system of first-order integro-differential equations
\be 
\begin{split} 
\frac{\id\vect{y}}{\id t} &= \vect{w} + \vect{A_u}(\vect{y}, t),\\
\frac{\id\vect{w}}{\id t} + \kappa \mu^{1/2}\frac{\id}{\id t}\left(\frac{1}{\sqrt{\pi}}\int_{t_0}^{t} \! \frac{\vect{w}(s)}{\sqrt{t - s}} \ \id s \right) +\mu \vect{w} &= - \vect{M_u}(\vect{y}, t)\vect{w} + \vect{B_u}(\vect{y}, t),\\
\vc y(t_0)=\vc y_0, &\quad  \vc w(t_0)=\vc w_0,
\end{split}\label{eq:MR_system_1}
\ee
with 
\begin{subequations} \label{eq:wterms}
\begin{alignat}{4}
\vect{w}(t)  =&\, \vect{v}(t) - \vect{u}(\vc y(t),t) - \frac{\gamma}{6}\mu^{-1} \Delta \vect{u}(\vc y(t),t), \label{eq:w} \\
\vect{A_u} =&\, \vect{u} + \frac{\gamma}{6}\mu^{-1}\Delta \vect{u}, \nonumber \\
\vect{B_u} =&\, \left(\frac{3R}{2} - 1\right)\left(\frac{\mbox D\vect{u}}{\mbox Dt} - \vect{g}\right) + \left(\frac{R}{20} - \frac{1}{6}\right)\gamma \mu^{-1} \frac{\mbox D\;}{\mbox Dt}\Delta \vect{u}\\
    &- \frac{\gamma}{6}\mu^{-1}\left[\nabla \vect{u} + \frac{\gamma}{6}\mu^{-1} \nabla \Delta \vect{u} \right] \Delta \vect{u}, \nonumber\\ 
     \vect{M_u} =&\, \nabla \vect{u} + \frac{\gamma}{6}\mu^{-1} \nabla \Delta \vect{u}.  \nonumber 
\end{alignat}
\end{subequations}

In deriving \eqref{eq:MR_system_1}, we used the identity 
$$\frac{\id}{\id t}\int_{t_0}^{t}\frac{\vect{w}(s)}{\sqrt{t - s}} \ \id s=\int_{t_0}^{t}\frac{\dot{\vc w}(s)}{\sqrt{t-s}}\id s+\frac{\vc w(t_0)}{\sqrt{t-t_0}},$$
obtained from carrying out the differentiation and then integrating by parts (see, e.g., \cite[Chapter 2]{podlubny1998fractional}).

The dimensionless parameters in \eqref{eq:wterms} are defined as
\be
R = \frac{2\rho_f}{\rho_f + 2\rho_p}, \qquad \mu = \frac{R}{\mbox{St}}, \qquad \kappa = \sqrt{\frac{9R}{2}}, \qquad \gamma = \frac{9R}{2\mbox{Re}},
\ee
where the Stokes (St) and the fluid Reynolds (Re) numbers are defined as
\be \label{eq:numbers}
\text{St} = \frac{2}{9}\left(\frac{a}{L}\right)^2\text{Re}, \qquad \text{Re} = \frac{UL}{\nu}.
\ee
Note that the vector fields $\vect{A_u},\vect{B_u}: \mathcal{D} \times \mathbb{R}^+ \rightarrow \mathbb{R}^n$ and the tensor field $\vect{M_u} : \mathcal{D} \times \mathbb{R}^+ \rightarrow \mathbb{R}^{n \times n}$ are known functions of the fluid velocity field $\vc u$. 

Equation \eqref{eq:w} defines a simple one-to-one correspondence between the particle velocity $\vc v$ and the variable $\vc w$. Once a solution $(\vc y,\vc w)$ of \eqref{eq:MR_system_1} is known, the particle velocity can readily be obtained as $\vc v(t)=\vc w(t)+\vc u(\vc y(t),t)+(\gamma\mu^{-1}/6)\Delta\vc u(\vc y(t),t)$. In the absence of the Fax\'en correction term $(\gamma\mu^{-1}/6)\Delta\vc u$, the variable $\vc w=\vc v-\vc u$ is the relative velocity between the particle and the fluid. 

The integral term in \eqref{eq:MR_system_1} is proportional to the Riemann-Liouville fractional derivative of order $1/2$, which is defined as
\be
\label{eq:riemann-liouville}
\frac{\id^{1/2}\vect{w}}{\id t^{1/2}} = \frac{\id}{\id t}\left(\frac{1}{\sqrt{\pi}}\int_{t_0}^{t} \! \frac{\vect{w}(s)}{\sqrt{t - s}} \ \id s\right),
\ee
with $t \geq t_0$ \cite{podlubny1998fractional}. Using this notation, we write the initial value problem \eqref{eq:MR_system_1} in the more compact form
\be 
\begin{split} 
\frac{\id\vect{y}}{\id t} &= \vect{w} + \vect{A_u}(\vect{y}, t),\\
\frac{\id\vect{w}}{\id t} + \kappa \mu^{1/2}\frac{\id^{1/2}\vc w}{\id t^{1/2}} +\mu \vect{w} &= - \vect{M_u}(\vect{y}, t)\vect{w} + \vect{B_u}(\vect{y}, t),\\
\vc y(t_0)=\vc y_0, &\quad  \vc w(t_0)=\vc w_0.
\end{split}\label{eq:MR_system_1comp}
\ee

\subsection{Set-up and assumptions}
\label{sec:preliminaries_assumptions}

We use $\abs{\cdot}$ to denote the Euclidean norm on $\mathbb{R}^{m}$. The induced operator norm of a square matrix acting on $\mathbb{R}^{m}$ is denoted by $\norm{\cdot}$. We denote the supremum norm of functions by $\norm{\cdot}_\infty$.

For future use, we also define the function space
\begin{equation}
X_{K}^{t,h} = \{f \in C\left(\left[t, t + h\right]; \mathbb{R}^m\right) : \norm{f}_\infty \leq K \}.
\end{equation}
Since $X_K^{t,h}$ is a closed subset of $C([t, t + h];\mathbb R^m)$, the metric space $(X^{t,h}_K, \norm{\cdot}_\infty)$ is a Banach space.

For the MR equation \eqref{eq:MR_system_1} (or its original form \eqref{eq:MR_original}) to make sense, the partial derivatives of the fluid velocity $\partial_x^{\alpha}\vc u(\vc x,t)$ and $\partial_t\partial_x^\beta\vc u(\vc x,t)$, with $|\alpha|\leq 3$ and $|\beta|\leq 2$ must exist. 

The Fax\'en corrections (the terms involving $\Delta\vc u$) are routinely neglected in practice \cite{IP_maxey87,balkovsky2001}. Upon neglecting the Fax\'en terms, the regularity assumption for the fluid velocity relaxes to the existence of the first order partial derivative with respect to space and time, i.e. $|\alpha|\leq 1$ and $\beta=0$. 

For proving the global existence and uniqueness of solutions of the MR equation, we need the above partial derivatives to be uniformly bounded and Lipschitz continuous in space and time. In particular, we assume the following.\vspace{10pt}

\begin{description}
\item[(H1)] The velocity field $\vc u(\vc x, t)$ is smooth enough such that the partial derivatives $\partial_x^{\alpha}\vc u$ with $|\alpha|\leq 3$ and the mixed partial derivatives $\partial_t\partial_x^\beta\vc u$ with $|\beta|\leq 2$ defined over the domain $\mathcal D \times R^+$ are uniformly bounded.
\end{description}

\begin{description}
\item[(H2)] The velocity field $\vc u(\vc x, t)$ is smooth enough such that the partial derivatives $\partial_x^{\alpha}\vc u$ with $|\alpha|\leq 3$ and the mixed partial derivatives $\partial_t\partial_x^\beta\vc u$ with $|\beta|\leq 2$ defined over the domain $\mathcal D \times R^+$ are uniformly Lipschitz continuous.
\end{description}

\begin{remark}
Neglecting the Fax\'en terms, assumptions (H1) and (H2) relax, respectively, to the uniform boundedness and uniform Lipschitz continuity of the fluid velocity $\vc u$ and acceleration $\mbox D\vc u/\mbox Dt$. 
\end{remark}

Assumption (H1) implies the existence of constants $L_A,L_B,L_M>0$ such that
\begin{equation} \label{eq:assumption1}
\norm{\vect{A_u}}_\infty\leq L_A,\quad\norm{\vect{B_u}}_\infty\leq L_B,\quad\norm{\vect{M_u}}_\infty\leq L_M.
\end{equation}
Assumption (H2), on the other hand, implies the existence of a constant $L_c>0$ such that
\begin{equation} \label{eq:assumption2}
\begin{split}
\abs{\vect{A_u}(\vect{y_1}, \tau) - \vect{A_u}(\vect{y_2}, \tau)} &\leq L_c\abs{\vect{y_1} - \vect{y_2}},\\
\abs{\vect{B_u}(\vect{y_1}, \tau) - \vect{B_u}(\vect{y_2}, \tau)} &\leq L_c\abs{\vect{y_1} - \vect{y_2}}, \\
\norm{\vect{M_u}(\vect{y_1}, \tau) - \vect{M_u}(\vect{y_2}, \tau)} &\leq L_c\abs{\vect{y_1} - \vect{y_2}},
\end{split}
\end{equation}
for all $\vect{y_1}$, $\vect{y_2} \in \mathcal{D}$ and all $\tau\in\mathbb R^+$. The supremum norms in \eqref{eq:assumption1} are taken over all $(\vc y,\tau)\in \mathcal D\times\mathbb R^+$.

\citet{MR_EUR} proved the following local existence and uniqueness result.
\begin{theorem}[Farazmand \& Haller, \cite{MR_EUR}]
\label{thm:local_e_u}
Assume that (H1) and (H2) hold. For any $(\vc y_0,\vc w_0)\in\mathcal D\times\mathbb R^n$, there exists $\Delta>0$ such that, over the time interval $[t_0, t_0+\Delta)$, the Maxey--Riley equation \eqref{eq:MR_system_1comp} has a unique solution $\left(\vect{y}(t),\vect{w}(t)\right)$ satisfying $\left(\vect{y}(t_0),\vect{w}(t_0)\right) = \left(\vect{y_0}, \vect{w_0}\right)$.
\end{theorem}

\subsection{The MR equation does not generate a dynamical system}\label{sec:notDS}
For ordinary differential equations, one may construct global solutions by continuation. In particular, given a local solution $(\vc y(t),\vc w(t))$ existing on a time interval $[t_0,t_0+\Delta_1)$, one shows that the solution does not blow up at $t=t_0+\Delta_1$. Then initializing the ordinary differential equation from time $t=t_0+\Delta_1$ with initial condition $(\vc y(t_0+\Delta_1),\vc w(t_0+\Delta_1))$, the local existence and uniqueness result is reapplied to show that the solution can be extended to an interval $[t_0,t_0+\Delta_1+\Delta_2)$. Repeating the above steps, the solution can be extended to a time interval $[t_0,t_0+\Delta_1+\Delta_2+\Delta_3+\cdots)$. Finally, one shows that the infinite series $\Delta_1+\Delta_2+\Delta_3+\cdots$ diverges and infers global existence and uniqueness.

This continuation argument assumes that the flow map $\vc F^t_{t_0}:(\vc y_0,\vc w_0)\mapsto (\vc y(t),\vc w(t))$ has the semi-group property $\vc F^t_{t_0}=\vc F^t_{t_1}\circ \vc F^{t_1}_{t_0}$ for all $t_0< t_1<t$. Due to the fractional derivative, however, the flow map of the MR equation \eqref{eq:MR_system_1comp} is not a semi-group. 

To see this, consider the solution $(\vc y(t),\vc w(t))$ starting from $(\vc y_0,\vc w_0)$ at time $t_0$. Due to the Basset history force (i.e., fractional derivative in \eqref{eq:MR_system_1comp}), the trajectory $(\vc y(t),\vc w(t))$ for $t>t_1$ is influenced by its entire past history. A trajectory initialized from $(\vc y(t_1),\vc w(t_1))$ is, however, ignorant of this history and therefore will follow a different path (see Fig. \ref{fig:MRDS}, for an illustration).
\begin{figure}[h!]
\centering\includegraphics[width=.4\textwidth]{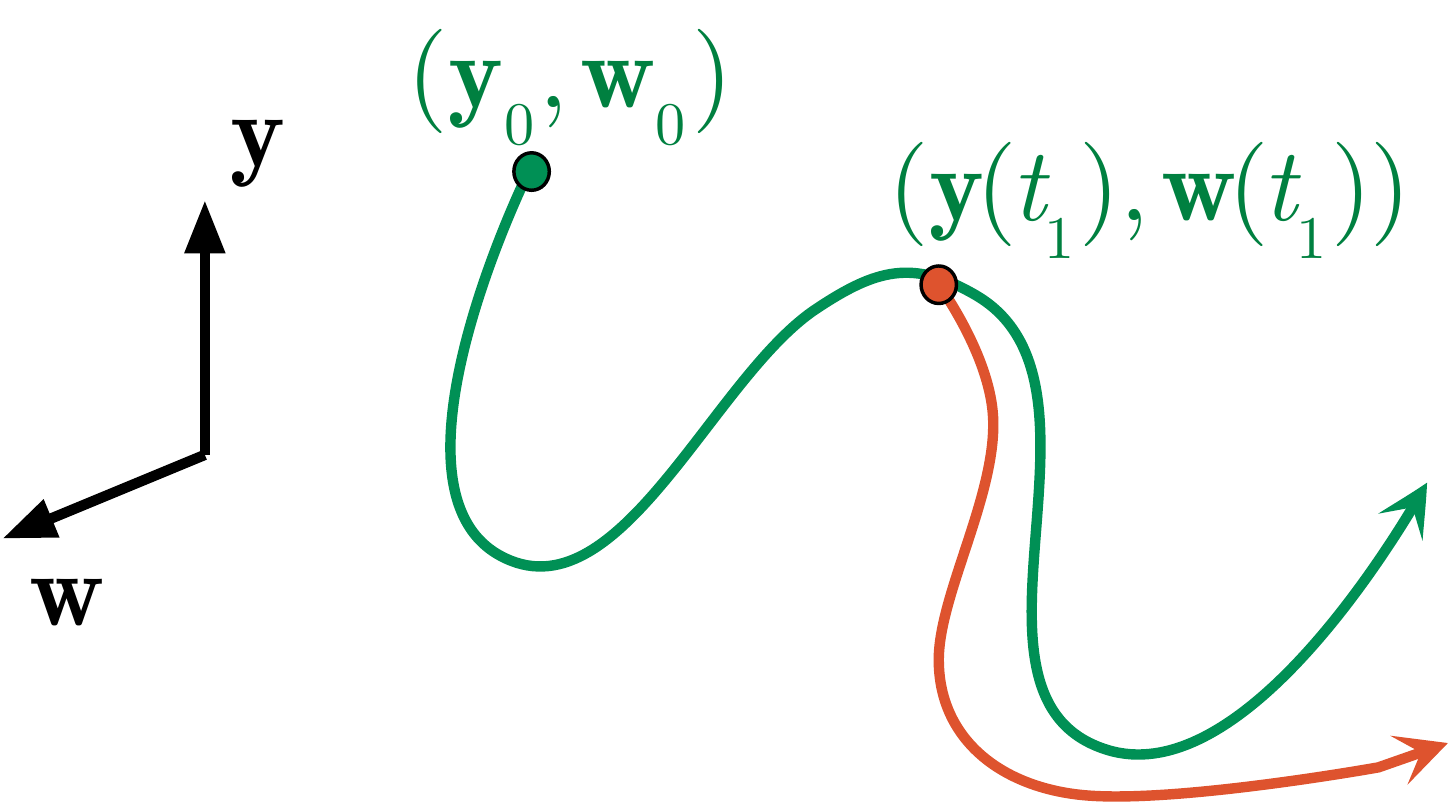}
\caption{A trajectory $(\vc y(t),\vc w(t))$ of the MR equation \eqref{eq:MR_system_1comp} initialized from $(\vc y_0,\vc w_0)$ and passing through $(\vc y(t_1),\vc w(t_1))$ at time $t_1$ (green curve). A trajectory initialized from $(\vc y(t_1),\vc w(t_1))$ at time $t_1$ (red curve) does not follow the trajectory $(\vc y(t),\vc w(t))$.}
\label{fig:MRDS}
\end{figure}

As a result, the usual continuation methods for ODEs do not apply here. In Section \ref{sec:cont}, we construct a particular continuation suitable for the MR equation. 

\subsection{Rescaling time}\label{sec:rescale}
We introduce a rescaling of time that further simplifies the forthcoming analysis. 
Dividing the $\vect{w}$ component of equation \eqref{eq:MR_system_1} by $\mu$ and letting $\epsilon := \frac{1}{\mu}$, we get
\be
\begin{split}
\frac{\id\vect{y}}{\id t} &= \vect{w} + \vect{A_u}(\vect{y}, t), \\
\epsilon\frac{\id\vect{w}}{\id t}+ \epsilon^{1/2}\kappa\frac{\id^{1/2}\vect{w}}{\id t^{1/2}} +\vect{w}&=  - \epsilon \vect{M_u}(\vect{y}, t)\vect{w} + \epsilon \vect{B_u}(\vect{y}, t),\\
\vc y(t_0)=\vc y_0, &\quad  \vc w(t_0)=\vc w_0.
\label{eq:MR_rescaled01}
\end{split}
\ee
Note that by \eqref{eq:numbers}, $\epsilon = \frac{\text{St}}{\text{R}} = \frac{2}{9R}\left(\frac{a}{L}\right)^2\text{Re}$. Since the MR equation is valid for small particles, i.e. $a\ll L$, we find that $\epsilon$ is a small parameter: $0\leq\epsilon\ll 1$. 

Rescaling time as $t = t_0+\epsilon \tau$, we have
\be 
\begin{split} 
\frac{\id\tilde{\vect{y}}}{\id\tau} = &\, \epsilon \left[\tilde{\vect{w}} + \tilde{\vect{A}}_{\vc u}(\tilde{\vect{y}}, \tau)\right],\\
\frac{\id\tilde{\vect{w}}}{\id\tau} + \kappa \frac{\id^{1/2}\tilde{\vect{w}}}{\id\tau^{1/2}} + \tilde{\vect{w}} = & \epsilon\left[ -\tilde{\vect{M}}_{\vc u}(\tilde{\vect{y}}, \tau)\tilde{\vect{w}} + \tilde{\vect{B}}_{\vc u}(\tilde{\vect{y}}, \tau)\right],\\
\tilde{\vc y}(0)=\vc y_0, &\quad  \tilde{\vc w}(0)=\vc w_0,
\end{split}\label{eq:MR_system_2}
\ee
where
\begin{subequations}
\begin{equation*}
\tilde{\vc y}(\tau)=\vc y(t_0+\epsilon\tau),\ \tilde{\vc w}(\tau)=\vc w(t_0+\epsilon\tau),
\end{equation*}
\begin{equation*}
\tilde{\vc A}_{\vc u}(\tilde{\vc y},\tau)=\vc{A_u}(\vc y,t_0+\epsilon\tau),\ \tilde{\vc B}_{\vc u}(\tilde{\vc y},\tau)=\vc{B_u}(\vc y,t_0+\epsilon\tau),\ \tilde{\vc M}_{\vc u}(\tilde{\vc y},\tau)=\vc{M_u}(\vc y,t_0+\epsilon\tau),
\end{equation*}
\end{subequations}
and
$$\frac{\id^{1/2}\tilde{\vect{w}}}{\id\tau^{1/2}} = \frac{\id}{\id\tau}\left(\frac{1}{\sqrt{\pi}}\int_{0}^{\tau} \frac{\tilde{\vect{w}}(s)}{\sqrt{\tau - s}} \ \id s\right).$$

The above rescaling of time has been previously used \cite{rubin1995_IP,mograbi2006,IP_haller08} for the asymptotic analysis of the MR equation without memory. While this rescaling is not necessary for the forthcoming results, it greatly simplifies the algebra. 

Note that a unique solution of the IVP \eqref{eq:MR_system_2} over the time interval $[0,\delta)$ exists if and only if the unscaled IVP \eqref{eq:MR_system_1comp} has a unique solution over the time interval $[t_0,t_0+\epsilon\delta)$. Therefore, in the following, we study the IVP \eqref{eq:MR_system_2}. For notational simplicity, we omit the tilde signs from all the variables.

\section{Asymptotic behavior}
\label{sec:properties}
\subsection{$\epsilon = 0$ limit}
\label{sec:properties_eps=0}
We start with the fictitious limit $\epsilon=0$. In this limit, $\vect{y}(\tau) = \vect{y_0}$ is constant for all times and $\vc w$ satisfies
\be \label{eq:MR_relaxation}
\frac{\id\vect{w}}{\id\tau} + \kappa \frac{\id^{1/2}\vect{w}}{\id\tau^{1/2}} + \vect{w} = 0,\quad \vc w(0)=\vc w_0,
\ee
a linear equation tractable by Laplace transforms \cite{mainardi_lecNotes,podlubny1998fractional}. This leads to the following result.

\begin{theorem}\label{thm:homogeneous_solution}
The general solution of \eqref{eq:MR_relaxation} is given by $\vect{w}(\tau;\vect{w_0}) =\psi_{\kappa}(\tau) \vect{w_0}$, where the positive, scalar function $\psi_\kappa:[0,\infty)\rightarrow\mathbb R^+$ has the following properties.
\begin{enumerate}
\item $\psi_\kappa$ is given by the inverse Laplace transform
\be \label{eq:l_relaxation}
\psi_\kappa(\tau) = \Lapl^{-1} \left[\frac{1}{\left(\sqrt{s} + \lambda_{+}\right)\left(\sqrt{s} + \lambda_{-}\right)}\right](\tau),
\ee
where
\begin{equation*}
\lambda_{\pm} = \frac{\kappa \pm \sqrt{\kappa^2 - 4}}{2}.
\end{equation*}
\item $\psi_\kappa$ obeys the asymptotic decay rate
\be
\psi_\kappa(\tau) \sim \frac{\kappa}{2\sqrt{\pi}}\tau^{-3/2} + \mathcal{O}\left(\tau^{-5/2}\right) \quad \text{as}\quad \tau \rightarrow \infty.
\ee
\item There is a differentiable function $\phi_\kappa:[0,\infty)\rightarrow\mathbb R^+$ such that $\psi_\kappa = - \phi_\kappa'$.
\item The functions $\psi_\kappa$ and $\phi_\kappa$ are smooth over $\in(0,\infty)$ and completely monotonic decreasing, i.e.,
\benn
(-1)^j\psi_\kappa^{(j)}(\tau)\geq 0,\ \ \ (-1)^j\phi_\kappa^{(j)}(\tau)\geq 0,\quad j=0,1,2,\cdots,\quad \forall\tau>0
\eenn
\item $\psi_\kappa(0) = 1$ and $\phi_\kappa(0) = 1$.  
\end{enumerate}
\end{theorem}
\begin{proof}
See Appendix \ref{app:calculation} for the proof of 1 and 2 and the explicit calculation of $\psi_\kappa$. For the proof of 3, 4 and 5, see the properties demonstrated for $u_\delta(t) (\psi_\kappa(\tau))$ and $u_0(t) (\phi_\kappa(\tau))$ in \cite[Section 4]{mainardi_lecNotes}. 
\end{proof}
Figure \ref{fig:phipsi} shows the functions $\phi_\kappa$ and $\psi_\kappa$ computed by numerically inverting their Laplace transforms. It follows from properties 2 and 3 from Theorem \ref{thm:homogeneous_solution} that $\phi_\kappa$ decays asymptotically as $\tau^{-1/2}$, as confirmed by the numerics. 
\begin{figure}
\centering
\includegraphics[width=.45\textwidth]{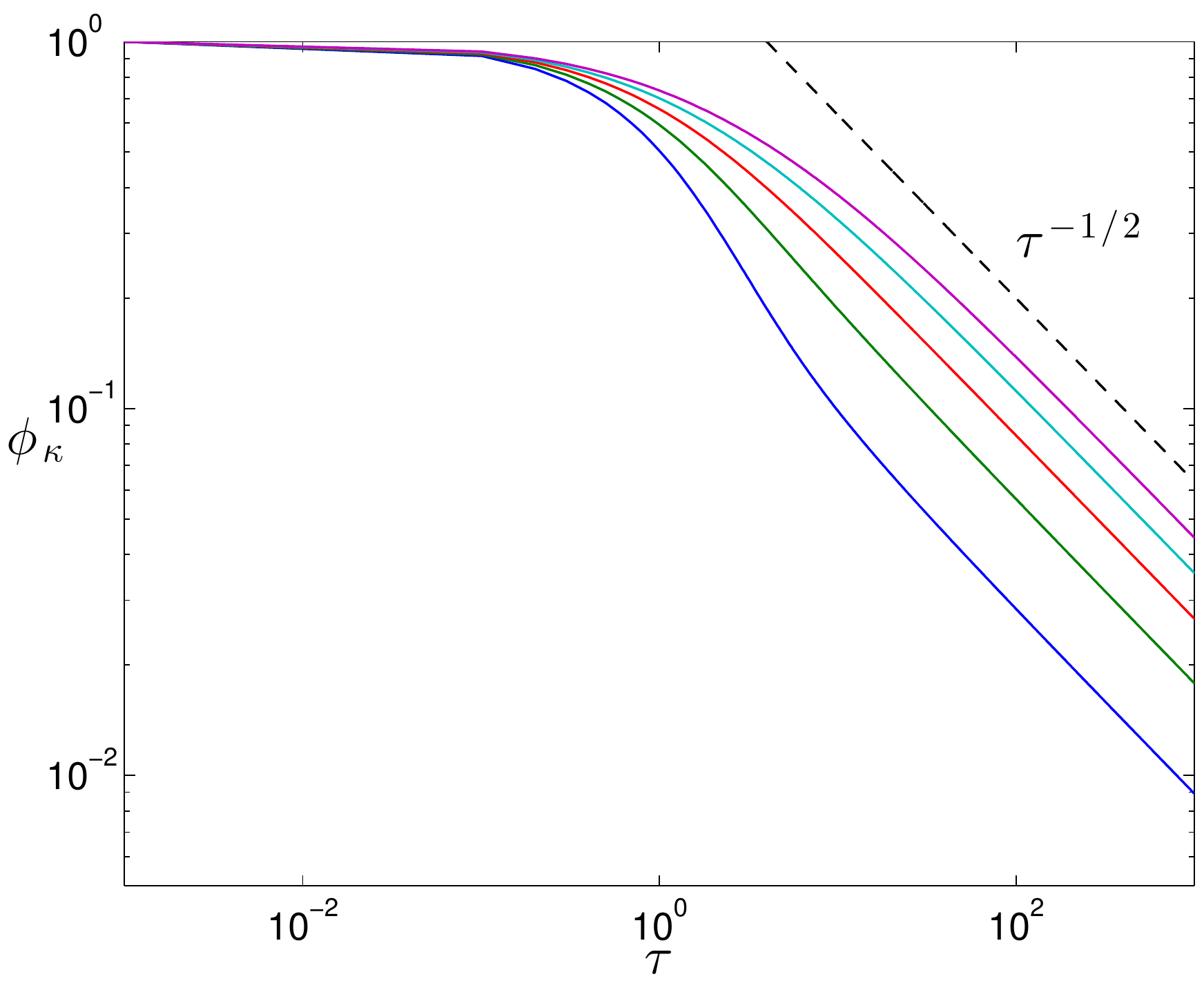}\hspace{.04\textwidth}
\includegraphics[width=.45\textwidth]{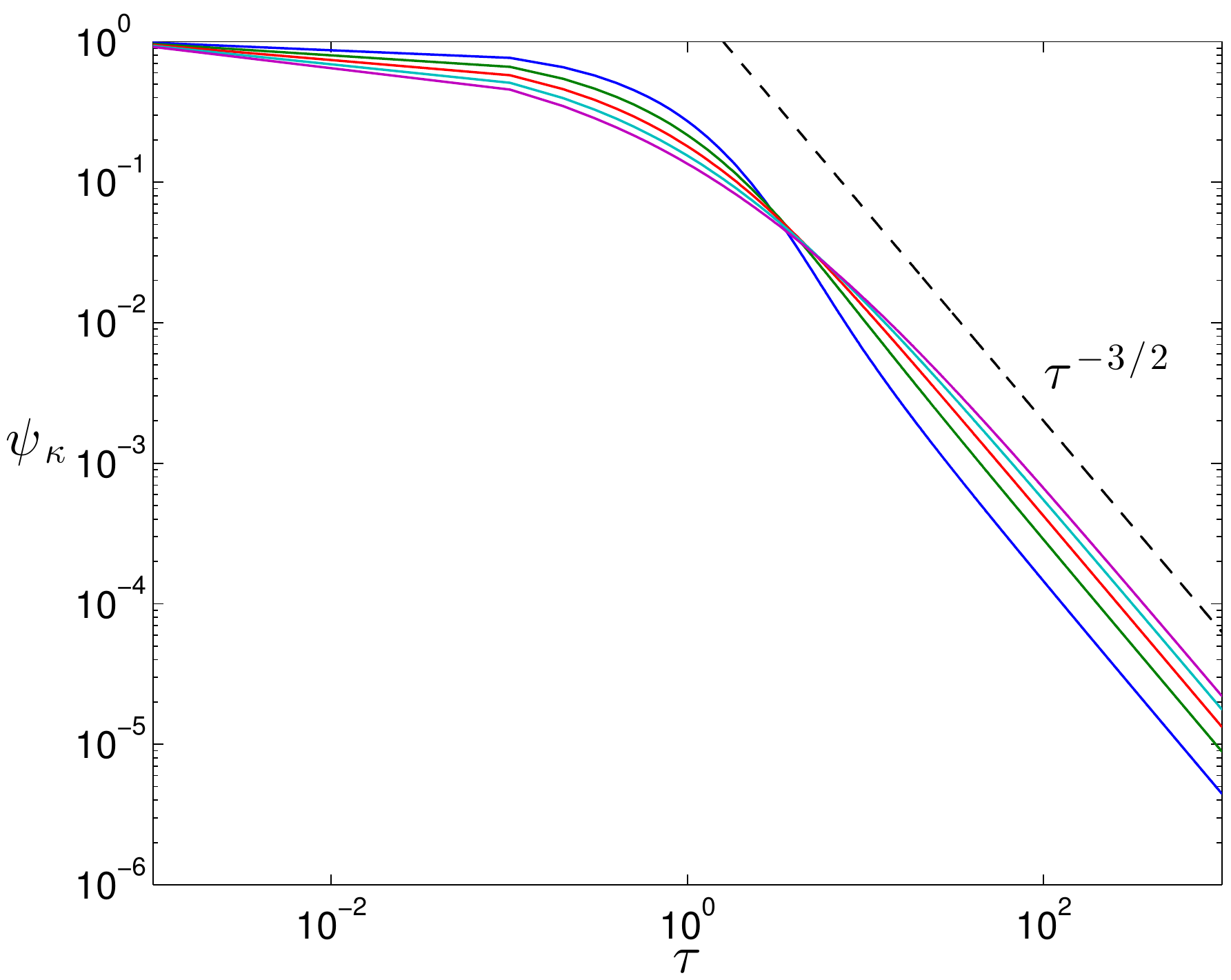}
\caption{The functions $\phi_\kappa$ and $\psi_\kappa=-\phi'_\kappa$.  The functions are evaluated for $\kappa=0.5$ (blue), $\kappa=1$ (green),
$\kappa=1.5$ (red), $\kappa=2$ (cyan) and $\kappa=2.5$ (magenta).}
\label{fig:phipsi}
\end{figure}

Since the properties of Theorem \ref{thm:homogeneous_solution} hold for any $\kappa>0$, we omit the dependence of $\psi_\kappa$ and $\phi_\kappa$ on $\kappa$ and write $\psi$ and $\phi$, respectively.

\subsection{$\epsilon > 0$ case}
\label{sec:properties_eps>0}
Now we analyze the general case of $\epsilon>0$, i.e., 
\be \label{eq:MR_relaxation_nonhomogeneous}
\begin{split}
\frac{\id\vect{y}}{\id\tau} &= \epsilon\left[\vect{w} + \vect{A_u}(\vect{y}, \tau)\right] \\
\frac{\id\vect{w}}{\id\tau} + \kappa \frac{\id^{1/2}\vect{w}}{\id\tau^{1/2}} + \vect{w} &= \epsilon\left[ -\vect{M_u}(\vect{y}, \tau)\vect{w} + \vect{B_u}(\vect{y}, \tau)\right],\\
\vc y(0)=\vc y_0, &\quad  \vc w(0)=\vc w_0,
\end{split}
\ee
which is equation \eqref{eq:MR_system_2} with tilde signs omitted. Solutions of \eqref{eq:MR_relaxation_nonhomogeneous} satisfy the integral equations
\be
\begin{split}
\label{eq:w_inteq}
\vect{y}(\tau) &= \vect{y_0} + \epsilon \int_{0}^{\tau} \! \vect{w}(s) + \vect{A_u}(\vect{y}(s), s) \ \id s, \\
\vect{w}(\tau) &= \psi(\tau)\vect{w_0} + \epsilon \int_{0}^{\tau} \! \psi(\tau - s)\left[-\vect{M_u}(\vect{y}(s), s)\vect{w}(s) + \vect{B_u}(\vect{y}(s), s)\right]\id s,
\end{split}
\ee
where $\psi(\tau)$ is given by \eqref{eq:l_relaxation} and satisfies the properties listed in Theorem \ref{thm:homogeneous_solution}. 

This integral equation is essentially a variation-of-constants formula. The $\vect{y}$-equation in \eqref{eq:w_inteq} is obtained by formal integration of the $\id\vect{y}/\id\tau$ equation of \eqref{eq:MR_relaxation_nonhomogeneous}. For the $\vc w$-equation, let $\vect{W}(s)$ denote the Laplace transform of $\vect{w}(\tau)$. Taking the Laplace transform of \eqref{eq:MR_relaxation_nonhomogeneous} yields
\benn \label{eq:l_nonhomogeneous}
\vect{W}(s) = \frac{\vect{w_0}}{\left(\sqrt{s} + \lambda_{+}\right)\left(\sqrt{s} + \lambda_{-}\right)} + \frac{\Lapl \left[-\vect{M_u}(\vect{y}(\tau), \tau)\vect{w}(\tau) + \vect{B_u}(\vect{y}(\tau), \tau)\right](s)}{\left(\sqrt{s} + \lambda_{+}\right)\left(\sqrt{s} + \lambda_{-}\right)}.
\eenn
Taking the inverse Laplace transform, we obtain the $\vc w$-component of equation \eqref{eq:w_inteq} where $\psi(\tau)$ is given by \eqref{eq:l_relaxation}.
\begin{definition}
A \emph{mild} (or weak) solution of the IVP \eqref{eq:MR_relaxation_nonhomogeneous} is a function $(\vc y,\vc w):[0,\delta)\rightarrow\mathbb R^{2n}$ that solves the integral equation \eqref{eq:w_inteq}. The existence time $\delta>0$ can potentially be infinity.
\end{definition}

Using the integral equation \eqref{eq:w_inteq}, we find an upper bound for $|\vc w(\tau;\vc y_0,\vc w_0)|$ and its asymptotic limit. 
\begin{theorem}
\label{thm:properties}
Assume that (H1) holds and $\epsilon < 1/L_M$. Let $(\vc y,\vc w):[0,\delta)\rightarrow\mathbb R^{2n}$ be a mild solution of \eqref{eq:MR_relaxation_nonhomogeneous} where $[0,\delta)$ is the maximal interval of existence of such solutions.
\begin{enumerate}[(i)]
\item An explicit envelope for $|\vc w(\tau;\vc y_0,\vc w_0)|$ is given by
\begin{equation}
|\vc w(\tau;\vc y_0,\vc w_0)|\leq \abs{\vect{w}_0}\left[\sum_{j=1}^{\infty}(\epsilon L_M)^{j-1}\psi^{*j}(\tau) \right] + \epsilon L_B\left(1-\phi(\tau)\right) +\frac{\epsilon^2 L_M L_B}{1 - \epsilon L_M},
\label{eq:envelope}
\end{equation}
where $\psi^{*j}$ is the $j$-fold convolution of $\psi$. Moreover, the series converges uniformly and is bounded for all $\tau$.
\item 
$\abs{\vect{w}(\tau;\vc y_0,\vc w_0)}$ is bounded for all $\tau\in[0,\delta)$. Specifically,
\be \label{eq:int_ineq}
\sup_{0\leq \tau<\delta}\abs{\vect{w}(\tau;\vc y_0,\vc w_0)} \leq \frac{\abs{\vect{w_0}} + \epsilon L_B}{1 - \epsilon L_M}.
\ee
\item If $\delta=\infty$, the asymptotic limit of $\vect w$ satisfies
\begin{equation}
\limsup_{\tau\rightarrow\infty}|\vc w(\tau;\vc y_0,\vc w_0)|\leq \frac{\epsilon L_B}{1-\epsilon L_M}.
\label{eq:asym_ub}
\end{equation}
\end{enumerate}
\end{theorem}
\begin{proof}
See Appendix \ref{app:properties}.
\end{proof}

In deriving the upper envelope \eqref{eq:envelope} and the subsequent upper bounds \eqref{eq:int_ineq} and \eqref{eq:asym_ub}, we have made several upper estimates. The natural question arising is how sharp these estimates are. In the following section, among other things, we show with a numerical example that these bounds are sharp by showing that they can be saturated. 

\subsection{Numerical verification}
We illustrate the results of Theorem \ref{thm:properties} with an example. For the fluid flow, we use the double gyre model of \citet{shadden05}. It is a two-dimensional velocity field with the stream function
\begin{equation}
\mathcal H(x,y,t)=A\sin(\pi f(x,t))\sin(\pi y),
\end{equation}
where
$$f(x,t)=\alpha\sin(\omega t)x^2+(1-2\alpha\sin(\omega t))x.$$
We let $A=0.1$, $\omega=\pi$ and $\alpha=0.01$. 
\begin{figure}[h!]
\centering
\subfloat[]{\includegraphics[width=.32\textwidth]{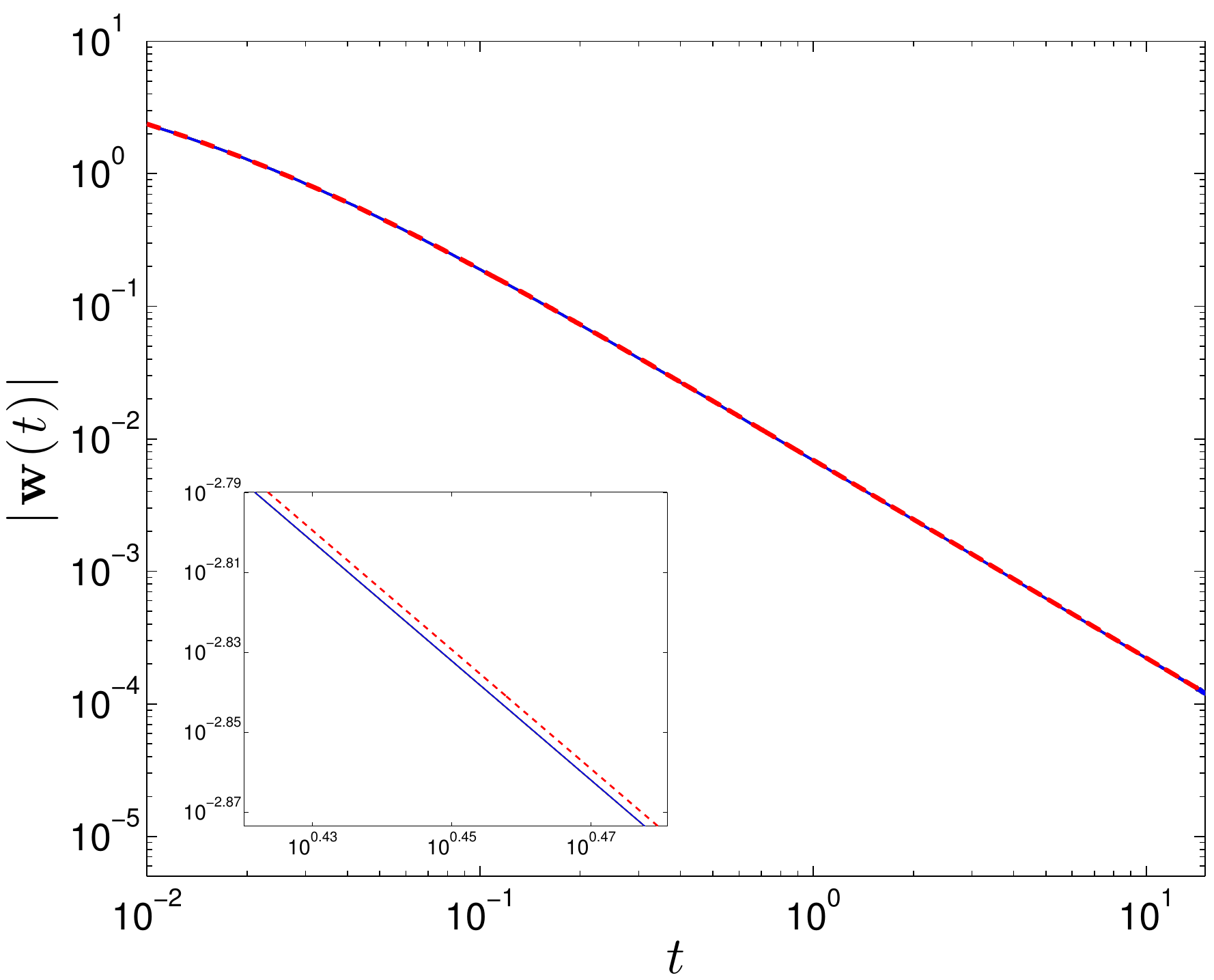}}
\subfloat[]{\includegraphics[width=.32\textwidth]{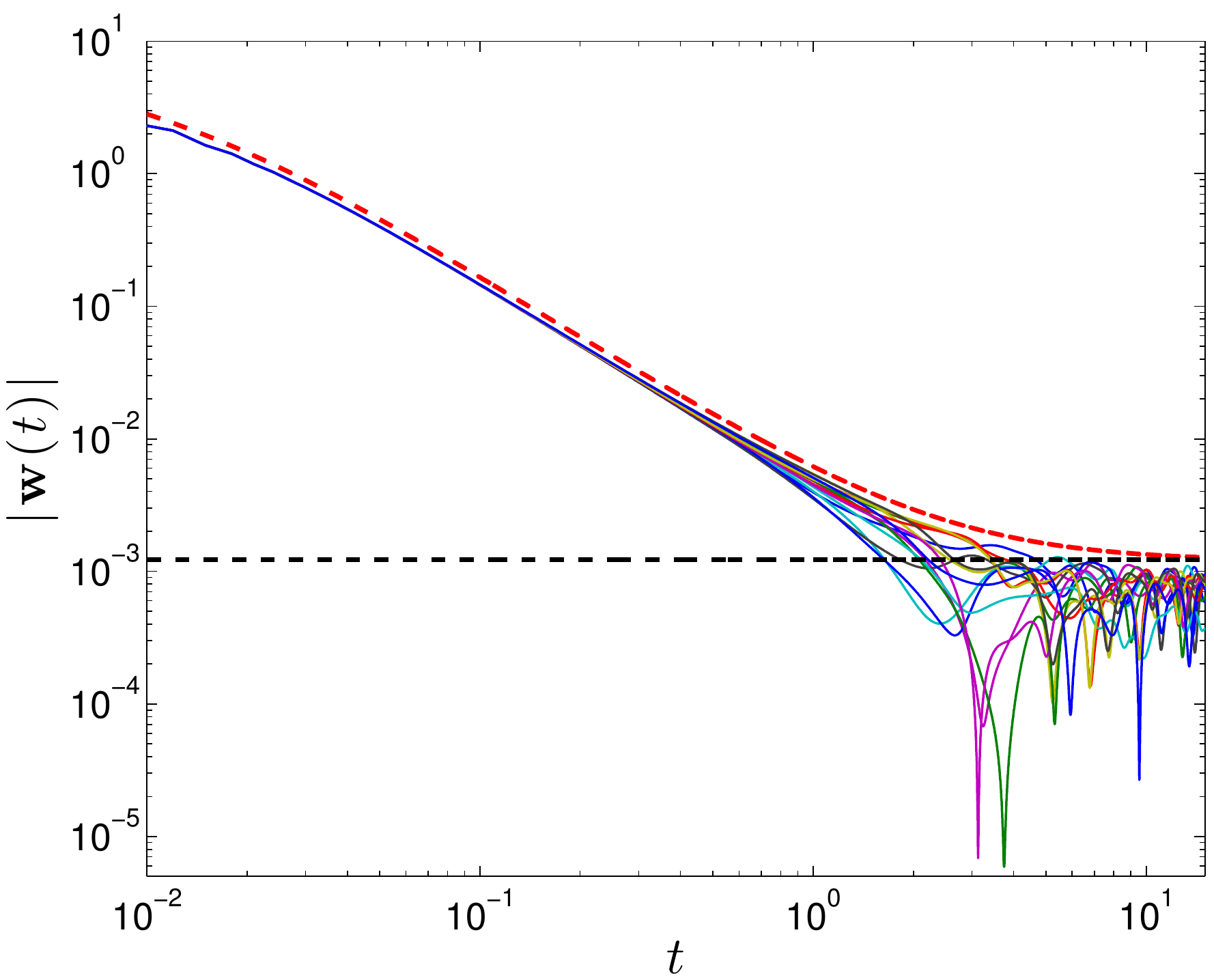}}
\subfloat[]{\includegraphics[width=.32\textwidth]{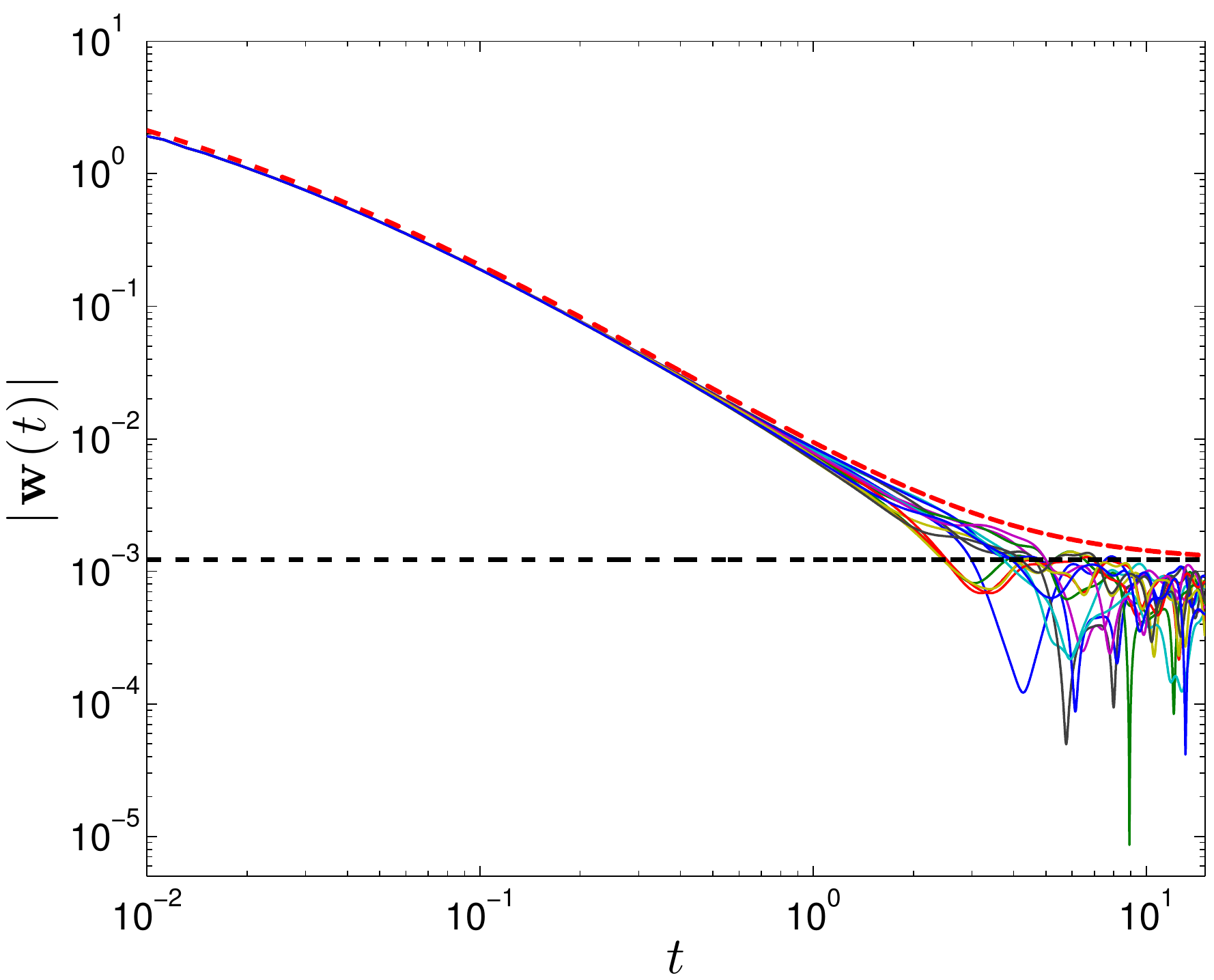}}
\caption{The decay of the relative velocity magnitude $|\vc w(t)|$ for $R=2/3$ (a), $R=1/3$ (b) and $R=1$ (c). The dashed red lines mark the analytic envelope from Theorem \ref{thm:properties} part (i). The dashed black lines mark the asymptotic upper bound of $|\vc w|$, i.e., $\epsilon L_B/(1-\epsilon L_M)$.}
\label{fig:dblgyre}
\end{figure}

The Hamiltonian $\mathcal H$ defines the velocity field $\vc u = (-\partial_y\mathcal H,\partial_x\mathcal H)^\top$ which we use to solve the initial value problem \eqref{eq:MR_system_1comp} using the numerical scheme developed in \cite{daitche2013advection}. We will neglect the Fax\'en corrections, such that $\vc A_{\vc u}=\vc u$, $\vc {B_u}=\left(\frac{3R}{2}-1\right)\frac{\mbox D\vc u}{\mbox Dt}$ and $\vc{M_u}=\bnabla\vc u$. 

For the parameters of the inertial particle, we let $\mbox{St}=R/100$ resulting in $\mu=100$ (or $\epsilon=0.01$). Three values of $R$ are considered here: $R=2/3$ (neutrally buoyant particle, $\rho_f=\rho_p$), $R=1/3$ (aerosol, $\rho_f<\rho_p$) and $R=1$ (bubble, $\rho_f>\rho_p$). In each case, we release $15$ trajectories with initial conditions $\vc y_0$ uniformly distributed in the domain $[0.2\times 1.8]\times[0.2,0.8]$ and identical initial relative velocities $\vc w_0=(10,10)^\top$. 

We take the most conservative choices of the upper bounds $L_B$ and $L_M$, i.e., $L_B=\|\vc{B_u}\|_\infty$ and $L_M=\|\vc{M_u}\|_\infty$. For the neutrally buoyant particle, i.e. $R=2/3$, $\vc {B_u}$ vanishes identically, resulting in $L_B=0$. The norm $\|\vc M_{\vc u}\|_\infty$ is, however, independent of $R$ and we have $L_M\simeq 1.4237$. Theorem \ref{thm:properties} therefore implies that for a neutrally buoyant particle, $|\vc w(t)|$ must decay to zero asymptotically which is in agreement with our numerical result (see Fig. \ref{fig:dblgyre}a). Physically, this implies that the inertial particle trajectory converges to a fluid trajectory. In the case of neutrally buoyant particles, the theoretical envelope and the numerical solutions almost coincide. A close-up view is shown in the inset of Fig. \ref{fig:dblgyre}a. 

Interestingly, for the neutrally buoyant particle, the evolution of the relative velocity magnitude $|\vc w|$ seems to be independent of the initial positions $\vc y_0$ as all $15$ curves coincide in Fig. \ref{fig:dblgyre}a.

For the bubble ($R=1$) and the aerosol ($R=1/3$), we have $L_B\simeq 0.1207$ and $L_M\simeq 1.4237$. The resulting envelope \eqref{eq:envelope} and the asymptotic upper bound $\epsilon L_B/(1-\epsilon L_M)$ are also shown (red and black dashed curves, respectively) which shows a perfect agreement with the numerical results. In plotting the envelopes, $\mathcal O(\epsilon^2)$-terms are neglected. The numerical solutions come very close to the analytic envelope of Theorem \ref{thm:properties} (part (i)), indicating the tightness of the estimates. 

The upper envelope \eqref{eq:envelope} depends on functions $\phi$ and $\psi$ which in turn depend on the parameter $\kappa=\sqrt{9R/2}$. The parameter $R$ is governed by the ratio between the particle density $\rho_p$ and the fluid density $\rho_f$. As this ratio varies the upper envelope also changes. Owing to the algebraic transient decay of $\phi$ and $\psi$ (see Fig. \ref{fig:phipsi}), however, the envelope exhibits an algebraic decay regardless of the value of $R$. Fig. \ref{fig:uppEnv_R} shows the behavior of the upper envelope (neglecting $\mathcal O(\epsilon^2)$-terms) for the double gyre parameters and various values of $R$. For neutrally buoyant particle ($R=2/3$) there is a monotonic decay with the algebraic rate $t^{-3/2}$. For other values of $R$ the envelope decays to the asymptotic upper bound. There is still a transient algebraic decay whose rate varies, depending on the parameter $R$, between $t^{-1.7}$ and $t^{-1.2}$. 
\begin{figure}
\centering
\includegraphics[width=.5\textwidth]{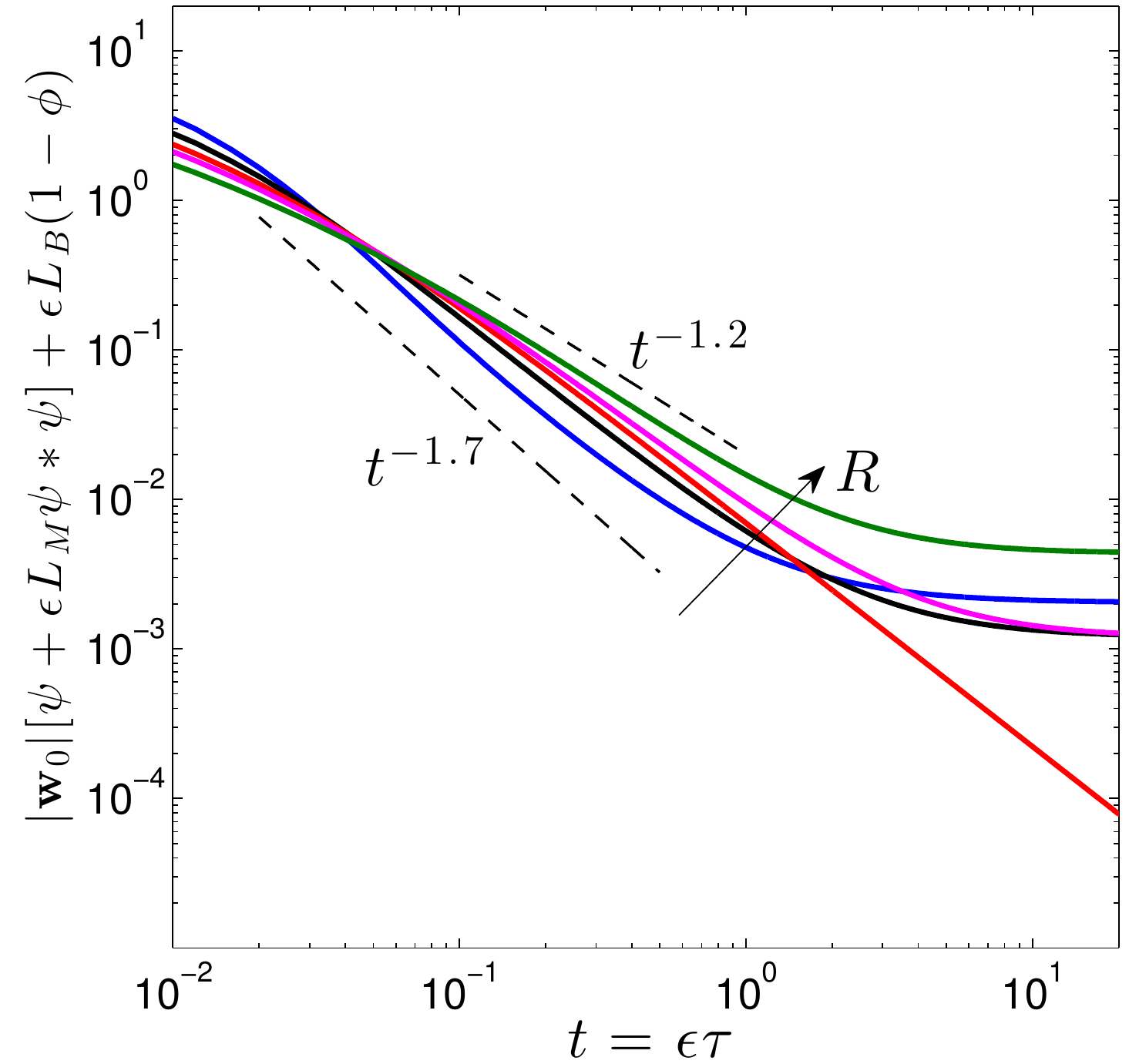}
\caption{The upper envelope \eqref{eq:envelope}, neglecting $\mathcal O(\epsilon^2)$-terms, for $R=1/10$ (blue), $R=1/3$ (black), $R=2/3$ (red), $R=1$ (magenta) and $R=19/10$ (green).}
\label{fig:uppEnv_R}
\end{figure}

\section{Global existence and uniqueness} \label{sec:global}
In this section, we prove the global existence and uniqueness of mild solutions of the Maxey--Riley equation \eqref{eq:MR_original}. In particular, we show that the equivalent reformulation \eqref{eq:MR_relaxation_nonhomogeneous} admits unique mild solutions for all times, i.e., the integral equations \eqref{eq:w_inteq} have a unique solution over $\mathbb R^+$.

The existence of a unique local solution follows from Theorem \ref{thm:local_e_u}. Specifically, there exists $\delta=\Delta/\epsilon>0$ such that the integral equation \eqref{eq:w_inteq} has a unique solution over the time interval $[0,\delta)$. Here, $\Delta$ is the same as the time window in Theorem \ref{thm:local_e_u} and the identity $\delta=\Delta/\epsilon$ follows from the rescaling $t=t_0+\epsilon\tau$ introduced in Section \ref{sec:rescale}. 

As discussed in Section \ref{sec:notDS}, the usual continuation methods used for ODEs does not apply to fractional order differential equations. Therefore, we construct a specific continuation method suitable for the MR equation, which is based on the continuation method presented in the work of Kou et al. \cite{kou2012existence} for a different class of fractional differential equations. We then show that this continuation can be repeated indefinitely to extend the solutions to the time interval $[0,\infty)$. Our approach can be summarized in the following steps.
\begin{description}
\item[Step 1.] Showing that the local solution of the integral equation \eqref{eq:w_inteq}, defined on $[0, \delta)$, is well defined at time $\tau = \delta$.
\item[Step 2.] Defining a suitable integral operator $\vc F$ over an appropriate Banach space whose fixed points extend the local solution of \eqref{eq:w_inteq} from $[0, \delta)$ to $[0, \delta + h)$, for a suitable constant $h > 0$.
\item[Step 3.] Showing that the operator $\vc F$ has at least one fixed point.
\item[Step 4.] Showing that this continuation is unique.
\item[Step 5.] Showing that one can repeat steps 1 to 4 indefinitely with the same continuation window $h$. That is the local solution of \eqref{eq:w_inteq} can be continued uniquely to $\mathbb R^+$.
\end{description}

The above steps prove the following global existence and uniqueness theorem.
\begin{theorem} \label{thm:global_mr}
Assume that (H1) and (H2) hold and $\epsilon<1/L_M$. The MR equation has unique, continuous, mild solutions. That is for any $(\vc y_0,\vc w_0)\in\mathbb R^{2n}$, there exists a unique, continuous function $(\vc y,\vc w):[0,\infty)\rightarrow\mathbb R^{2n}$ satisfying \eqref{eq:w_inteq} and $(\vc y(0),\vc w(0))=(\vc y_0,\vc w_0)$.
\end{theorem}

\subsection{Continuation of the local solution}\label{sec:cont}
Let's denote the local solution of the MR equation, whose existence and uniqueness is guaranteed by Theorem \ref{thm:local_e_u}, by $\vc z_{loc}=(\vc y_{loc},\vc w_{loc})$. We first begin by showing that this local solution defined on $[0, \delta)$ is well defined at $\tau = \delta$.

\begin{lemma} \label{lemma:WD}
The local solution $\vect{z}_{loc}:[0,\delta)\rightarrow \mathbb R^{2n}$ of the MR equation is well-defined at $\tau = \delta$ and the limit $\lim_{\tau \to \delta^-} \vect{z}_{loc}(\tau)$ is given by
\begin{equation} \label{eq:MR_LIMIT}
\vect{z}_{loc}(\delta) = \left( 
  \begin{array}{lr}
    \vect{y_0} + \epsilon \int_{0}^{\delta} \! \vect{w}_{loc}(s) + \vect{A_u}(\vect{y}_{loc}(s),s) \ \id s \\
    \psi(\delta)\vect{w_0} + \epsilon \int_{0}^{\delta} \! \psi(\delta - s)\left[-\vect{M_u}(\vect{y}_{loc}(s), s)\vect{w}_{loc}(s) + \vect{B_u}(\vect{y}_{loc}(s), s)\right] \ \id s
  \end{array}
\right).
\end{equation}
\end{lemma}
\begin{proof}
See Appendix \ref{app:lemma_WD}.
\end{proof}

Let $(\vc y_{loc},\vc w_{loc}):[0,\delta)\rightarrow\mathbb R^{2n}$ be the local solution of \eqref{eq:w_inteq} whose existence and uniqueness is guaranteed by Theorem \ref{thm:local_e_u}. Define
\begin{subequations}
\begin{equation}
\vc y(\tau) = \indf_{[0,\delta)}(\tau)\vc y_{loc}(\tau)+\indf_{[\delta,\delta+h)}(\tau) \bxi(\tau),
\end{equation} 
\begin{equation}
\vc w(\tau) = \indf_{[0,\delta)}(\tau)\vc w_{loc}(\tau)+\indf_{[\delta,\delta+h)}(\tau) \betta(\tau),
\end{equation}
\label{eq:ext}
\end{subequations}
where $\indf_A:\mathbb R\rightarrow\{0,1\}$ is the indicator function of the set $A\subset\mathbb R$. Note that for $\tau\in[0,\delta)$, $(\vc y,\vc w)$ coincides with the local solution $(\vc y_{loc},\vc w_{loc})$. Assuming $(\vc y,\vc w)$ is a continuation of this local solution to $[0,\delta+h)$, upon substitution in \eqref{eq:w_inteq}, we have
\begin{equation}
\begin{split}
\betta(\tau) =& \vect{y_0} + \epsilon \int_{0}^{\delta} \! \vect{w}_{loc}(s) + \vect{A_u}(\vect{y}_{loc}(s), s) \ \id s + \epsilon \int_{\delta}^{\tau} \! \betta(s) + \vect{A_u}(\bxi(s), s) \ \id s, \\
\bxi(\tau) =& \psi(\tau)\vect{w_0} + \epsilon \int_{0}^{\delta} \! \psi(\tau - s)\left[-\vect{M_u}(\vect{y}_{loc}(s), s)\vect{w}_{loc}(s) + \vect{B_u}(\vect{y}_{loc}(s), s)\right] \ \id s  \\ 
          &+ \epsilon \int_{\delta}^{\tau} \! \psi(\tau - s)\left[-\vect{M_u}(\bxi(s), s)\betta(s) + \vect{B_u}(\bxi(s), s)\right] \ \id s,
\end{split}
\label{eq:MR_EXTENSION}
\end{equation}
for $\tau\in[\delta,\delta+h)$. 

Therefore, $(\vc y,\vc w)$ solves the integral equation \eqref{eq:w_inteq} and hence is a mild solution of the MR equation if and only if the integral equation \eqref{eq:MR_EXTENSION} has a solution. To show that such a solution exists, we solve the following fixed point problem. Let $\vect{\Phi} = (\bxi,\betta) \in X_K^{\delta,h}$.  Define the operator $\vc F : X_K^{\delta,h} \rightarrow C([\delta, \delta + h); \mathbb{R}^{2n})$ by
\begin{equation}
\label{eq:Fphi}
\left(\vc F\vect{\Phi}\right)(\tau) = \vect{\Phi_0}\left(\tau\right) + \left(
  \begin{array}{lr}
    \epsilon \int_{\delta}^{\tau} \! \betta(s) + \vect{A_u}(\bxi(s), s) \ \id s \\
    \epsilon \int_{\delta}^{\tau} \! \psi(\tau - s)\left[-\vect{M_u}(\bxi(s), s)\betta(s) + \vect{B_u}(\bxi(s), s)\right] \ \id s
  \end{array}
\right),
\end{equation}
where
\begin{equation}
\label{eq:phi0}
\vect{\Phi_0}\left(\tau\right) = \left(
  \begin{array}{lr}
    \vect{y_0} + \epsilon \int_{0}^{\delta} \! \vect{w}_{loc}(s) + \vect{A_u}(\vect{y}_{loc}(s), s) \ \id s \\
    \psi(\tau)\vect{w_0} + \epsilon \int_{0}^{\delta} \! \psi(\tau - s)\left[-\vect{M_u}(\vect{y}_{loc}(s), s)\vect{w}_{loc}(s) + \vect{B_u}(\vect{y}_{loc}(s), s)\right] \ \id s
  \end{array}
\right).
\end{equation}
Note that $\vect\Phi_0$ depends only on the local solution $(\vect{y}_{loc}, \vect{w}_{loc})$ of the Maxey--Riley equation and hence is independent of $\vect\Phi$. We show that the operator $\vc F$ maps $X_K^{\delta,h}$ to itself (with $K$ and $h$ to be determined) and has a unique fixed point. 


\subsection{Existence of the continuation}\label{sec:exist}
\begin{prop}\label{lem:existence_fp}
Assume that (H1) holds. There exist constants $h,K>0$ such that the operator $\vc F$ defined in \eqref{eq:Fphi} maps $X_K^{\delta,h}$ to itself and has at least one fixed point.
\end{prop}
\begin{proof}
For any $h,K>0$ and $\vect\Phi\in\X$ the function $\vc F\vect\Phi:[\delta,\delta+h)\rightarrow \mathbb R^{2n}$ is clearly continuous, i.e. $\vc F\vect\Phi\in C([\delta,\delta+h);\mathbb R^{2n}$. We choose $h,K>0$ such that $\vc F\vect\Phi\in\X$, i.e., $\|\vc F\vect\Phi\|_\infty\leq K$. To this end, note that for any $h>0$ and $\tau\in[\delta,\delta+h)$, we have
\begin{equation*}
\begin{split}
\abs{(\vc F\vect{\Phi})(\tau)} 	 \leq & \abs{\vect{\Phi_0}(\tau)} + \epsilon\int_{\delta}^{\delta + h} \! \abs{\betta(s)} + \abs{\vect{A_u}(\bxi(s), s)} \ \id s \\
										& +     \epsilon\int_{\delta}^{\delta + h} \! \psi(\tau-s)\left[\abs{\vect{M_u}(\bxi(s), s)\betta(s)} + \abs{\vect{B_u}(\bxi(s), s)}\right] \ \id s.
\end{split}
\end{equation*}
Take the supremum over $\tau\in[\delta,\delta+h)$ and use the bounds on $\norm{\vect{M_u}}_\infty$, $\norm{\vect{B_u}}_\infty$, $\norm{\vect{A_u}}_\infty$, $\abs{\vect{w}(\tau)}$, $\norm{\psi}_\infty$, and $\norm{\betta}_\infty$ to get 
\begin{equation*}
\begin{split}
\norm{\vc F\vect{\Phi}}_\infty &\leq \norm{\vect{\Phi_0}}_\infty + \epsilon\int_{\delta}^{\delta + h} \! \norm{\betta}_\infty + \norm{\vect{A_u}}_\infty \ \id s \\
                                   & + \epsilon\int_{\delta}^{\delta + h} \! \big[\norm{\vect{M_u}}_{\infty}\norm{\betta}_{\infty} + \norm{\vect{B_u}}_\infty\big] \ \id s \\
																	& \leq \norm{\vect{\Phi_0}}_\infty + \epsilon h\left(K + L_A\right) + \epsilon h \left(L_M K + L_B\right).
\end{split}
\end{equation*}
For
\benn
h \leq \frac{1}{\epsilon\left(L_M + 1\right)},
\eenn
we have
\benn
\norm{\vc F\vect{\Phi}}_\infty \leq \norm{\vect{\Phi_0}}_\infty + \frac{K}{2} + \frac{L_B + L_A}{2\left(L_M + 1\right)}.
\eenn
Since $\vect\Phi_0:[0,\infty)\rightarrow\mathbb R^{2n}$ is a continuous function, there exists $0<K'<\infty$ such that
\begin{equation*}
\norm{\vect{\Phi_0}}_\infty:=\sup_{\delta\leq\tau<\delta+h}|\vect\Phi_0(\tau)| = K'
\end{equation*}

Choosing
\benn
K \geq \left[K'  + \frac{L_B + L_A}{2\left(L_M + 1\right)}\right],
\eenn
we have $\norm{\vc F\vect{\Phi}}_\infty \leq K$.

In short, with any $h,K>0$ satisfying
\be
h \leq \frac{1}{\epsilon\left(L_M + 1\right)},\quad K = K' + \frac{L_B + L_A}{2\left(L_M + 1\right)},
\label{eq:hKbound}
\ee
the operator $\vc F$ maps $\X$ to itself.

To prove the existence of a fixed point for the operator $\vc F:X_K^{\delta,h}\rightarrow X_K^{\delta,h}$, we use Schauder's fixed point theorem:

\begin{theorem}[Schauder's Fixed Point Theorem] \label{thm:schauder}
Let $X$ be a real Banach space, $D \subset X$ nonempty, closed, bounded, and convex. Let $\mathcal F : D \rightarrow D$ be a continuous, compact operator. Then $\mathcal F$ has a fixed point. 
\end{theorem}

The space $\X$ is nonempty, closed, bounded and convex. Therefore, to apply Schauder's fixed point theorem, it remains to show that $\vc F:\X\rightarrow\X$ is continuous and compact. For this, we need the following lemma. 
\begin{lemma} \label{lem:continuous} 
The operator $\vc F$ is continuous and maps $\X$ to a family of equi-continuous functions in $\X$.
\end{lemma}
\begin{proof}
The proof of the continuity of $\vc F:\X\rightarrow \X$ is straightforward and is therefore omitted here. We prove the equicontinuity of its range in Appendix \ref{app:lemma_continuous}.
\end{proof}

By Arzela-Ascoli theorem, therefore, the operator $\vc F:\X\rightarrow\X$ is compact. Hence, $\vc F$ satisfies all the conditions of Schauder's theorem and has at least one fixed point. This concludes the proof of Proposition \ref{lem:existence_fp}.
\end{proof}

\subsection{Uniqueness of the continuation}\label{sec:unique}
We now show that the continuation constructed in sections \ref{sec:cont} and \ref{sec:exist} is unique. 
\begin{prop}
Assume that (H1) and (H2) hold and $\epsilon<1/L_M$. There exists $h>0$ such that the continuation \eqref{eq:ext} of the local solution of the MR equation is unique. 
\end{prop}

\begin{proof}
Suppose $\left(\vect{y}_1,\vect{w}_1\right)$ and $\left(\vect{y}_2,\vect{w}_2\right)$ are two different continuations of the local solution of \eqref{eq:w_inteq} from $[0, \delta)$ to $[\delta, \delta + h)$. That is
\begin{equation*}
\vc y_1(\tau) = \indf_{[0,\delta)}(\tau)\vc y_{loc}(\tau)+\indf_{[\delta,\delta+h)}(\tau) \bxi_1(\tau),\; \vc w_1(\tau) = \indf_{[0,\delta)}(\tau)\vc w_{loc}(\tau)+\indf_{[\delta,\delta+h)}(\tau) \betta_1(\tau),
\end{equation*}
and
\begin{equation*}
\vc y_2(\tau) = \indf_{[0,\delta)}(\tau)\vc y_{loc}(\tau)+\indf_{[\delta,\delta+h)}(\tau) \bxi_2(\tau),\; \vc w_2(\tau) = \indf_{[0,\delta)}(\tau)\vc w_{loc}(\tau)+\indf_{[\delta,\delta+h)}(\tau) \betta_2(\tau),
\end{equation*}
where, as discussed in Section \ref{sec:cont}, $(\xi_i,\eta_i)$ solve the integral equations 
\begin{equation}
\begin{pmatrix}
\xi_i(\tau)\\ \eta_i(\tau)
\end{pmatrix}=\vect\Phi_0(\tau)+\epsilon
\left(
  \begin{array}{lr}
     \int_{\delta}^{\tau} \! \betta_i(s) + \vect{A_u}(\bxi_i(s), s) \ \id s \\
     \int_{\delta}^{\tau} \! \psi(\tau - s)\left[-\vect{M_u}(\bxi_i(s), s)\betta_i(s) + \vect{B_u}(\bxi_i(s), s)\right] \ \id s
  \end{array}
\right),
\end{equation}
for $i \in \{1, 2\}$. 

Define $\vect\Phi_i=(\bxi_i,\betta_i)$ and bound $|\vect\Phi_1-\vect\Phi_2|$ by
\begin{equation}
\begin{split} \label{eq:u_inequality}
|\vect\Phi_1(\tau)-\vect\Phi_2(\tau)| \leq & \epsilon\int_{\delta}^{\delta + h} \! \abs{\betta_1(s) - \betta_2(s)} + \abs{\vect{A_u}(\bxi_1(s), s) - \vect{A_u}(\bxi_2(s), s)} \ \id s \\
                              & +  \epsilon\int_{\delta}^{\delta + h} \! \abs{\psi(\tau - s)} \left(\abs{\vect{M_u}(\bxi_1(s), s)(\betta_1(s) - \betta_2(s))} \right. \\
															&\quad\quad +  \left. \abs{\betta_2(s)} \abs{\vect{M_u}(\bxi_1(s), s) - \vect{M_u}(\bxi_2(s), s)}\right) \ \id s,\\
															& +  \epsilon\int_{\delta}^{\delta + h} \! \abs{\psi(\tau - s)}\abs{B_u(\bxi_1(s), s) - \vect{B_u}(\bxi_2(s), s)} \ \id s,
\end{split}
\end{equation}
where we wrote $\abs{\vect{M_u}(\bxi_1(s), s)\betta_1(s) - \vect{M_u}(\bxi_2(s), s)\betta_2(s)}$ as
\benn
\abs{\vect{M_u}(\bxi_1(s), s)(\betta_1(s) - \betta_2(s)) + \left(\vect{M_u}(\bxi_1(s), s) - \vect{M_u}(\bxi_2(s), s)\right)\betta_2(s)}.
\eenn

Since $(\vc y_i,\vc w_i)$ solves the MR equation $[0, \delta + h)$, inequality \eqref{eq:int_ineq} applies and we have 
$$\|\betta_i\|_\infty:=\sup_{\delta\leq\tau<\delta+h}|\betta_i(\tau)|\leq\sup_{0\leq\tau<\delta+h}|\vc w_i(\tau)|\leq\frac{\abs{\vect{w_0}} + \epsilon L_B}{1 - \epsilon L_M},\quad i\in\{0,1\}.$$
Taking the supremum over $\tau\in[\delta,\delta+h)$ on both sides of \eqref{eq:u_inequality} and using the above upper bound on $\|\betta_i\|_\infty$, we get
\begin{equation*}
\begin{split}
\norm{\vect\Phi_1-\vect\Phi_2}_\infty & \leq \epsilon h \left[L_c \norm{\betta_1 - \betta_2}_\infty + L_c \norm{\bxi_1 - \bxi_2}_\infty \right] \\  
&  +   \epsilon h \left[L_M \norm{\betta_1 - \betta_2}_\infty + L_c \left(\frac{\abs{\vect{w_0}} + \epsilon L_B}{1 - \epsilon L_M}\right) \norm{\bxi_1 - \xi_2}_\infty + L_c \norm{\bxi_1 - \bxi_2}_\infty \right], \\
             & \leq 2h\epsilon \left[3L_c + L_M + L_c \left(\frac{\abs{\vect{w_0}} + \epsilon L_B}{1 - \epsilon L_M}\right)\right] \norm{\vect\Phi_1-\vect\Phi_2}_\infty.
\end{split}
\end{equation*}

Taking $h>0$ small enough, one obtains $\|\vect\Phi_1-\vect\Phi_2\|_\infty\leq\frac{1}{2}\|\vect\Phi_1-\vect\Phi_2\|_\infty$ which, in turn, implies the uniqueness of the solution: $\vect\Phi_1=\vect\Phi_2$. The time window $h$ can for instance be chosen as
\begin{equation} \label{eq:h}
h = \frac{1}{2} \min \left(\frac{1}{\epsilon(L_M + 1)}, \frac{1}{2\epsilon \left[3L_c + L_M + L_c \left(\frac{\abs{\vect{w_0}} + \epsilon L_B}{1 - \epsilon L_M}\right)\right]}\right),
\end{equation}
which also respects the inequality \eqref{eq:hKbound}. With this $h$, therefore, the continuation \eqref{eq:ext} is unique.
\end{proof}

\begin{remark}
The above analysis is a contraction mapping argument. It is, therefore, tempting to use the Banach fixed point theorem (instead of the Schauder's fixed point theorem) in order to obtain the existence and uniqueness of the continuation \eqref{eq:ext} at once. The Banach fixed point theorem, however, does not apply here. This is because in proving the above contraction property, we made use of inequality \eqref{eq:int_ineq} which applies to the mild solutions of the MR equation. As a result, it was necessary to show the existence of continuation \eqref{eq:ext} first. Otherwise, inequality \eqref{eq:int_ineq} does not apply and the estimates used in the above contraction argument fail.
\end{remark}
\vspace{10pt}

So far we have proved the existence of a unique mild solution to the MR equation over the time interval $[0,\delta+h)$ with $h$ given in \eqref{eq:h}. The steps taken in sections \ref{sec:cont}, \ref{sec:exist} and \ref{sec:unique} can be applied to this extended local solution to prove the existence and uniqueness of a mild solution over the time interval $[0,\delta+2h)$. This is because the continuation window $h$ is independent of the constants $K$ and $\delta$ from the Banach space $X_K^{\delta,h}$.

Applying this argument repeatedly extends the mild solution of the Maxey--Riley equation from its local interval of existence and uniqueness $[0, \delta)$ to $[0, \delta + nh]$, for any $n \in \mathbb{N}$. Thus the solution can be extended uniquely to $[0, \infty)$. This proves Theorem \ref{thm:global_mr}. 

\section{Summary and discussion}
Motivated by the recent observations on the relevance of the memory effects on inertial particle dynamics, we have derived global existence and asympotic decay results for the Maxey--Riley equation in the presence of the Basset--Boussinesq memory term. This memory term, a fractional derivative of order $1/2$ \cite{daitche2013advection,MR_EUR}, greatly complicates the analytical and numerical treatment of the equation. While the behavior of the solutions has been well-understood in the absence of the memory term \cite{rubin1995_IP,mograbi2006,IP_haller08,sapsis2010clustering}, no global analytic results have been available for the full equation with memory.

We have proved that the solutions converge asymptotically to a trapping region where the particle velocity is $\mathcal O(\epsilon)$-close to the fluid velocity. Here, $\epsilon$ is proportional to $(a/L)^2$ where $a$ is the particle radius and 
$L$ is the characteristic length-scale of the fluid flow. This result holds for $0<\epsilon\ll 1$ small enough which translates into $a\ll L$ (See Theorem \ref{thm:properties}, for the exact statement of the assumption). This assumption is not restrictive since the MR equation is only valid under the very same condition $a\ll L$ \cite{MR}. 

We also derived an upper envelope for the transient dynamics. This envelope exhibits an algebraic decay to the asymptotic state, hence confirming the numerical observations of \cite{daitche2011memory,guseva2013influence,Daitche_NJP} in a more general framework. We showed with an example that this envelope can be saturated and therefore our upper estimates are sharp. 

Upon neglecting the memory term, the convergence to the asymptotic limit is exponential \cite{rubin1995_IP,mograbi2006,IP_haller08}. Therefore, the Basset--Boussinesq memory fundamentally alters the behavior of the inertial particles and cannot be readily neglected. From a mathematical point of view, the memory term also fundamentally changes the structure of the equation. In the absence of memory, the Maxey--Riley equation is an ordinary differential equation, generating a dynamical system. The memory term turns the equation into an integro-differential equation that does not generate a dynamical system.

Our asymptotic results are only applicable if the Maxey--Riley equation possesses global solutions. Because of the particular coupling and nonlinearity of the equation, available results on fractional-order differential equations do not guarantee the existence and uniqueness of global solutions to the Maxey--Riley equation. To this end, we have included here the first proof of the existence and uniqueness of global solutions to the Maxey--Riley equation. As already pointed out by \cite{MR_EUR}, the particle velocity is only guaranteed to be continuous for all times.

\section*{Acknowledgment}
We would like to thank Anton Daitche for his help with implementing the numerical scheme of \cite{daitche2013advection}.

\begin{appendices}
\addcontentsline{toc}{part}{\appendixname}
\section{Proof of Theorem \ref{thm:homogeneous_solution}}
\label{app:calculation}

Consider the fractional differential equation
\be
\label{eq:relaxation}
\frac{d\vect{w}}{d\tau} + \kappa\frac{d^{1/2}\vect{w}}{d\tau^{1/2}} + \vect{w} = 0, 
\ee
with $\vect{w}(0) = \vect{w}_0$ as initial condition. Let $\vect{W}(s) = \left(\Lapl\left[\vect{w}\right]\right)\left(s\right)$ denote the Laplace transform of $\vect{w}\left(\tau\right)$. Since
\benn
\left(\Lapl \left[\frac{d\vect{w}}{d\tau}\right]\right)(s) = s\vect{W}(s) - \vect{w}_0
\eenn
and
\benn
\left(\Lapl \left[\frac{1}{\sqrt{\tau}}\right]\right)(s) = \sqrt{\frac{\pi}{s}},
\eenn
the Laplace transform of the Riemann-Liouville derivative in \eqref{eq:relaxation} has the expression
\begin{equation*}
\begin{split}
\left(\Lapl \left[\frac{d^{1/2}\vect{w}}{d\tau^{1/2}}\right]\right)(s) &= \frac{1}{\sqrt{\pi}} \left(\Lapl \left[\int_{0}^{\tau} \! \frac{d\vect{w}}{d\tau}\frac{1}{\sqrt{\tau - s}} \ \id s\right]\right)(s) + \frac{1}{\sqrt\pi}\left(\Lapl \left[\frac{\vect{w}_0}{\sqrt{\tau}}\right]\right)(s), \\
 &= \frac{1}{\sqrt{\pi}}\left(\Lapl \left[\frac{d\vect{w}}{d\tau}\right] \right)(s)\left(\Lapl \left[\frac{1}{\sqrt{\tau}}\right] \right)(s) + \frac{\vect{w}_0}{\sqrt{s}}, \\
 &= \left(s\vect{W}(s) - \vect{w}_0\right)\frac{1}{\sqrt{s}} + \frac{\vect{w}_0}{\sqrt{s}}, \\
 &= \sqrt{s}\vect{W}(s),
\end{split}
\end{equation*}
where we used the identity
\benn
\frac{\id}{\id \tau}\int_{0}^{\tau}\frac{\vect{w}(s)}{\sqrt{\tau - s}} \ \id s = \int_{0}^{\tau}\frac{\id \vect{w}}{\id \tau}\frac{1}{\sqrt{\tau-s}} \ \id s + \frac{\vc w(0)}{\sqrt{\tau}}.
\eenn
Now we use the Laplace transform on \eqref{eq:relaxation} and solve for $\vect{W}(s)$ to get
\benn
\vect{W}\left(s\right) = \frac{\vect{w}_0}{s + \kappa\sqrt{s} + 1}.
\eenn
The denominator can be factorized as
\benn
\label{eq:l_relaxation2}
\vect{W}\left(s\right) = \frac{\vect{w}_0}{\left(\sqrt{s} + \lambda_{+}\right)\left(\sqrt{s} + \lambda_{-}\right)},
\eenn
where
\benn
\label{eq:lambda_pm}
\lambda_{\pm} = \frac{\left(\kappa \pm \sqrt{\kappa^2 - 4}\right)}{2}.
\eenn
Hence the general solution of \eqref{eq:relaxation} is
\be
\vect{w}(\tau; \vect{w}_0) = \vect{w}_0 \left(\Lapl^{-1} \left[\frac{1}{\left(\sqrt{s} + \lambda_{+}\right)\left(\sqrt{s} + \lambda_{-}\right)}\right]\right)\left(\tau\right)
\ee
The function $\vect{w}(\tau;\vc w_0)$ is proportional to the Mittag-Leffler function of order $1/2$, which is defined as
\be
\label{eq:mittag}
E_{1/2}\left(-z\right) = e^{z^2}\erfc{z}
\ee
for any complex number $z \in \mathbb{C}$ (see, e.g., \cite{bateman1955higher}, Section 18.1). The Laplace transform of $E_{1/2}$ is given by (see \cite{haubold2011mittag}, Eq. 11.13):
\be
\label{eq:l_rule}
\left(\Lapl \left[E_{1/2}\left(-a\sqrt{z}\right)\right]\right)(s) = \frac{1}{\sqrt{s}\left(\sqrt{s} + a\right)}
\ee
for any $a \in \mathbb{C}$.

To study the behavior of $E_{1/2}\left(-z\right)$ as $z \rightarrow \infty$, we will make use of the asymptotic expansion of the complementary error function (\cite{Abramowitz+Stegun}, Eq. 7.1.23):
\be
\erfc{z} \sim \frac{e^{-z^2}}{z\sqrt{\pi}}\left(1 - \frac{1}{2z^2} + \frac{3}{4z^4} + \mathcal O\left(\frac{1}{z^6}\right)\right).
\ee
Substituing in \eqref{eq:mittag}, we obtain
\be
\label{eq:mittag_asymptotic}
E_{1/2}\left(-z\right) \sim \frac{1}{z\sqrt{\pi}}\left(1 - \frac{1}{2z^2} + \frac{3}{z^4} + \mathcal O\left(\frac{1}{z^6}\right)\right).
\ee
The asymptotic expansion of $\erfc{z}$ is valid only if $\abs{\operatorname{arg}\left(z\right)} < \frac{3\pi}{4}$ \cite{Abramowitz+Stegun}. It also diverges for any finite value of z; its sole purpose is to give the rate of decay as $z \rightarrow \infty$.

The general solution will depend on whether the discriminant of $\lambda_{\pm}$, i.e. $\kappa^2 - 4$, is positive, zero, or negative.

\subsection{Case 1: $\kappa > 2$ (i.e., $R>16/9$)}

We have
\benn
\label{eq:l_case1}
\vect{W}\left(s\right) = \frac{\vect{w}_0}{\left(\sqrt{s} + \lambda_{+}\right)\left(\sqrt{s} + \lambda_{-}\right)}
\eenn
or, after some algebra,
\benn
\label{eq:l_case1_2}
\vect{W}\left(s\right) = \frac{\vect{w}_0}{\lambda_{+} - \lambda_{-}}\left[\frac{\lambda_+}{\sqrt{s}\left(\sqrt{s} + \lambda_+\right)} - \frac{\lambda_-}{\sqrt{s}\left(\sqrt{s} + \lambda_-\right)}\right].
\eenn
Invert the two terms in the above expression with the rule \eqref{eq:l_rule} to get
\be
\label{eq:case1sol}
\vect{w}(\tau; \vect{w}_0) = \frac{\vect{w}_0}{\lambda_{+} - \lambda_{-}}\left[\lambda_+E_{1/2}\left(-\lambda_+\sqrt{\tau}\right) - \lambda_-E_{1/2}\left(-\lambda_-\sqrt{\tau}\right)\right].
\ee

Since $\kappa - \sqrt{\kappa^2-4}$ is always greater than zero, we can use the asymptotic expansion \eqref{eq:mittag_asymptotic} to find that in the limit $\tau \rightarrow \infty$,
\be
\begin{split}
\vect{w}(\tau; \vect{w}_0) &\sim \frac{\vect{w}_0}{\lambda_+ - \lambda_-}\left[\frac{1}{\sqrt{\pi \tau}}\left(1 - \frac{1}{2\lambda_{+}^{2}\tau}\right) \right. \\
& \left. - \frac{1}{\sqrt{\pi \tau}}\left(1 - \frac{1}{2\lambda_{-}^{2}\tau}\right) + \mathcal O\left(\tau^{-5/2}\right) \right], \\
&\sim \frac{\vect{w}_0}{2\sqrt{\pi}\left(\lambda_+ - \lambda_-\right)} \left(\frac{\lambda_{+}^{2} - \lambda_{-}^{2}}{\lambda_{+}^{2}\lambda_{-}^{2}}\right)\tau^{-3/2} + \mathcal O\left(\tau^{-5/2}\right), \\
&\sim \left(\frac{\kappa \vect{w}_0}{2\sqrt{\pi}}\right)\tau^{-3/2} + \mathcal O\left(\tau^{-5/2}\right),
\end{split}
\ee
where we used that $\lambda_+ + \lambda_- = \kappa$ and $\lambda_+\lambda_- = 1$.

\subsection{Case 2: $\kappa = 2$ (i.e., $R=16/9$)}

We have
\be
\label{eq:l_case2}
\vect{W}(s) = \frac{\vect{w}_0}{\left(\sqrt{s} + 1\right)^2}
\ee
or, after a bit of algebra,
\begin{equation}
\begin{split} \label{eq:l_case2_2}
\vect{W}(s) &= \vect{w}_0\left(\frac{1}{\sqrt{s}\left(\sqrt{s} + 1\right)} - \frac{1}{\sqrt{s}\left(\sqrt{s} + 1\right)^2}\right) \\
     &= \vect{w}_0\left(\frac{1}{\sqrt{s}\left(\sqrt{s} + 1\right)} + 2 \frac{d}{ds}\left(\frac{1}{\sqrt{s} + 1}\right)\right).
\end{split}
\end{equation}
We can invert the first term in \eqref{eq:l_case2_2} with \eqref{eq:l_rule}. The second term can be inverted by using the Laplace transforms \cite[Equations A.27, A.28, and A.35.]{mainardi_lecNotes}
\be
\label{eq:l_rule_2}
\left(\Lapl \left[\frac{1}{\sqrt{\pi \tau}} - E_{1/2}\left(-\sqrt{\tau}\right)\right]\right)(s) = \frac{1}{\sqrt{s} + 1}
\ee
and
\be
\left(\Lapl \left[-\tau f(\tau)\right]\right)(s) = \frac{d}{ds}(\Lapl \left[f(\tau)\right])(s).
\ee
Thus the inverse Laplace transform of \eqref{eq:l_case2} is
\be
\vect{w}(\tau; \vect{w}_0) = \vect{w}_0\left[E_{1/2}\left(-\sqrt{\tau}\right)\left(1 + 2 \tau\right) - \frac{2\sqrt{\tau}}{\sqrt{\pi}}\right].
\ee

With the asymptotic expansion \eqref{eq:mittag_asymptotic} we find that in the limit $\tau \rightarrow \infty$,
\be
\begin{split}
\vect{w}(\tau; \vect{w}_0) &\sim \vect{w}_0 \left[\frac{1}{\sqrt{\pi \tau}}\left(1 - \frac{1}{2 \tau } + \frac{3}{4 \tau^2} + \mathcal O\left(\tau^{-3}\right)\right) \right. \\
&+ \left. \frac{2\sqrt{\tau}}{\sqrt{\pi}}\left(1 - \frac{1}{2 \tau} + \frac{3}{4 \tau^2} + \mathcal O\left(\tau^{-3}\right)\right) - \frac{2\sqrt{\tau}}{\sqrt{\pi}}\right] \\
&\sim \left(\frac{\vect{w}_0}{\sqrt{\pi}}\right)\tau^{-3/2} + \mathcal O \left(\tau^{-5/2}\right) \\
\end{split}
\ee

\subsection{Case 3: $0 < \kappa < 2$ (i.e., $R<16/9$)}

We have
\benn
\label{eq:l_case3}
\vect{W}\left(s\right) = \frac{\vect{w}_0}{\left(\sqrt{s} + \lambda_{+}\right)\left(\sqrt{s} + \lambda_{-}\right)}
\eenn
This is the same Laplace transform as in the case $\kappa > 2$, except that $\lambda_+$ and $\lambda_-$ are now complex conjugate numbers. The inverse Laplace transform is the same as \eqref{eq:case1sol}:
\be \label{eq:sol_case3_1}
\vect{w}(\tau; \vect{w}_0) = \frac{\vect{w}_0}{\lambda_{+} - \lambda_{-}}\left[\lambda_+E_{1/2}\left(-\lambda_+\sqrt{\tau}\right) - \lambda_-E_{1/2}\left(-\lambda_-\sqrt{\tau}\right)\right].
\ee
The quotients
\benn
\frac{\lambda_+}{\lambda_+ - \lambda_-} = \frac{1}{2}\left(1 - i \frac{\kappa}{\sqrt{4 - \kappa^2}}\right)
\eenn
and
\benn
-\frac{\lambda_-}{\lambda_+ - \lambda_-} = \frac{1}{2}\left(1 + i \frac{\kappa}{\sqrt{4 - \kappa^2}}\right)
\eenn
in \eqref{eq:sol_case3_1} are also complex conjugates. Since $\left(e^{\overline{w}}\right) = \overline{\left(e^w\right)}$ and $\erfc{\overline{w}} = \overline{\erfc{w}}$ for every $w \in \mathbb{C}$, it follows also that $E_{1/2}\left(\overline{w}\right) = \overline{E_{1/2}\left(w\right)}$. Thus
\benn
\vect{w}(\tau; \vect{w_0}) = \vect{w_0}\left[\left(\frac{\lambda_+}{\lambda_+ - \lambda_-}E_{1/2}\left(-\lambda_+\sqrt{\tau}\right)\right) + \overline{\left(\frac{\lambda_+}{\lambda_+ - \lambda_-}E_{1/2}\left(-\lambda_+\sqrt{\tau}\right)\right)} \right]
\eenn
or simply twice the real part of $w(\tau; w_0)$.
\be \label{eq:sol_case3_2}
\begin{split}
\vect{w}(\tau; \vect{w}_0) &= 2\vect{w}_0\operatorname{Re}\left(\frac{\lambda_+}{\lambda_+ - \lambda_-}E_{1/2}(-\lambda_+\sqrt{\tau})\right), \\
        &= 2\vect{w}_0\left[\operatorname{Re}\left(\frac{\lambda_+}{\lambda_+ - \lambda_-}\right)\operatorname{Re}\left(E_{1/2}\left(-\lambda_+\sqrt{\tau}\right)\right) \right. \\
				&+ \left. \operatorname{Im}\left(\frac{\lambda_+}{\lambda_+ - \lambda_-}\right)\operatorname{Im}\left(E_{1/2}\left(-\lambda_+\sqrt{\tau}\right)\right) \right].
\end{split}
\ee
It is possible to further simplify \eqref{eq:sol_case3_1}. It turns out that the Mittag-Leffler function $E_{1/2}\left(-z\right)$ may be written as (\cite{NIST:DLMF}, section 7.19)
\be
E_{1/2}\left(-z\right) = \sqrt{\frac{4t}{\pi}}\left[U\left(x, t\right) + i V\left(x, t\right)\right],
\ee
where
\be
U\left(x,t\right) = \frac{1}{\sqrt{4\pi t}} \int_{-\infty}^{\infty} \! \frac{e^{-\left(x + s\right)^2/\left(4t\right)}}{1+s^2} \ \id s,
\ee
\be
V\left(x,t\right) = \frac{1}{\sqrt{4\pi t}} \int_{-\infty}^{\infty} \! \frac{se^{-\left(x + s\right)^2/\left(4t\right)}}{1+s^2} \ \id s,
\ee
$z = \frac{1-ix}{2\sqrt{t}}$, $x \in \mathbb{R}$, and $t > 0$. The functions $U\left(x,t\right)$ and $V\left(x,t\right)$ are known as the Voigt functions (\cite{NIST:DLMF}, section 7.19). If we set
\benn
z = \frac{1-ix}{2\sqrt{t}} = \lambda_+ \sqrt{\tau} = \left(\frac{\kappa}{2} + i \frac{\sqrt{4 - \kappa^2}}{2}\right)\sqrt{\tau},
\eenn
then we can solve for $x$ and $t$ to get
\benn
t=\frac{1}{\kappa^2\tau}
\eenn
and
\benn
x = -\frac{\sqrt{4 - \kappa^2}}{\kappa}.
\eenn
Thus
\be
\begin{split}
E_{1/2}\left(-\lambda_+\sqrt{\tau}\right) &= \frac{2}{\kappa\sqrt{\pi \tau}}\left[U\left(-\frac{\sqrt{4-\kappa^2}}{\kappa},\frac{1}{\kappa^2\tau}\right) \right. \\
&- \left. i V\left(-\frac{\sqrt{4-\kappa^2}}{\kappa},\frac{1}{\kappa^2\tau}\right)\right].
\end{split}
\ee
Hence \eqref{eq:sol_case3_2} can be written as
\be
\begin{split}
\vect{w}\left(\tau; \vect{w}_0\right) &= \frac{2\vect{w}_0}{\kappa\sqrt{\pi \tau}} \left[U\left(-\frac{\sqrt{4-\kappa^2}}{a},\frac{1}{\kappa^2\tau}\right) \right. \\ 
&- \left. \frac{\kappa}{\sqrt{4 - \kappa^2}}V\left(-\frac{\sqrt{4-\kappa^2}}{\kappa},\frac{1}{\kappa^2\tau}\right) \right].
\end{split}
\ee

For the asymptotic behaviour of $\vect{w}(\tau; \vect{w}_0)$ as $\tau \rightarrow \infty$, we can repeat the steps as in the case $\kappa > 2$ and obtain
\be
\vect{w}(\tau; \vect{w}_0) \sim \left(\frac{\kappa \vect{w}_0}{2\sqrt{\pi}}\right)\tau^{-3/2} + \mathcal O\left(\tau^{-5/2}\right).
\ee
This asymptotic expansion, however, is justified only if $\abs{\operatorname{arg}\left(\lambda_+\sqrt{\tau}\right)}$ and $\abs{\operatorname{arg}\left(\lambda_+\sqrt{\tau}\right)}$ are smaller than $\frac{3\pi}{4}$. Since $\lambda_{\pm} = \left(\kappa \pm i\sqrt{4 - \kappa^2}\right)/2$ we see that this will be the case whenever $\kappa > 0$, since then $0 < \operatorname{arg}\left(\lambda_+\sqrt{\tau}\right) < \frac{\pi}{2}$ and $ -\frac{\pi}{2} < \operatorname{arg}\left(\lambda_-\sqrt{\tau}\right) < 0$ (to see this, note that the two complex numbers $\lambda_+$ and $\lambda_-$ lie to the right of the imaginary axis, so that the argument cannot be greater than $\pi / 2$). Note that since $\kappa=\sqrt{9R/2}$, the required condition $\kappa>0$ is always satisfied. 

\section{Proof of Theorem \ref{thm:properties}}
\label{app:properties}

We will use the following Gronwall-type inequality.
\begin{lemma}[Chu \& Metcalf \cite{chu1967gronwall}]
Let the functions $\alpha, \beta: \mathbb R^+\rightarrow\mathbb R$ be continuous and the function $K(\tau,s)$ be continuous and nonnegative for $0\leq s\leq \tau$. If 
$$\alpha(\tau)\leq \beta (\tau)+\int_{0}^\tau K(\tau,s)\alpha(s) \ \id s,$$
then
$$\alpha(\tau)\leq \beta (\tau)+\int_{0}^\tau H(\tau,s)\beta(s)\ \id s,$$
where $H(\tau,s)=\sum_{j=1}^\infty K_j(\tau,s)$, $K_1(\tau,s)=K(\tau,s)$ and 
$$K_j(\tau,s)=\int_{s}^{\tau}K_{j-1}(\tau,\xi)K(\xi,s)\ \id\xi,\ \ \ j\geq 2.$$
\label{lem:gronwall}
\end{lemma}
\begin{cor}
If $K(\tau,s)=k(\tau-s)$, then one can show that $K_j(\tau,s)=k_j(\tau-s)$ where
$$k_j(\tau)=(k\ast k\ast\cdots\ast k)(\tau),$$
where the convolution is $j$-fold. As a result, $H(\tau,s)=h(\tau-s)$ where
$$h(\tau)=\sum_{j=1}^{\infty}k_j(\tau).$$
\end{cor}
\begin{proof}
We prove $K_2(\tau,s)=k\ast k(\tau-s)$. The rest follows similarly by induction. 
\begin{align*}
K_2(\tau,s) & :=\int_{s}^{\tau}K(\tau,\xi)K(\xi,s) \ \id\xi\\
            & =\int_{s}^{\tau}k(\tau-\xi)k(\xi-s)\ \id\xi\\
            & = \int_0^{\tau-s}k(\tau-s-\eta)k(\eta)\ \id\eta\\
            & = k\ast k(\tau-s)=:k_2(\tau-s),
\end{align*}
where we used the change of variable $\eta=\xi-s$.
\end{proof}

\begin{proof}[Proof of Theorem \ref{thm:properties}]
It follows from the integral equation \eqref{eq:w_inteq} that
\begin{equation}
\abs{\vect{w}(\tau;\vect{y}_0,\vect{w}_0)}\leq \psi(\tau)\abs{\vect{w}_0}+\epsilon L_B\left(1 - \phi(\tau)\right)+\epsilon L_M\int_{0}^\tau \psi (\tau-s)\abs{\vect{w}(s;\vect{y}_0,\vect{w}_0)}\ \id s
\end{equation}
where $\tau \in [0, \delta)$.
Using Lemma \ref{lem:gronwall} with $\alpha(\tau) = \abs{\vect{w}(\tau;\vect{y}_0,\vect{w}_0)}$, $\beta(\tau) = \psi(\tau)\abs{\vect{w}_0} + \epsilon L_B\left(1-\phi(\tau)\right)$ and $K(\tau,s)=\epsilon L_M\psi(\tau-s)$, we get
\begin{align}
\abs{\vect{w}(\tau;\vect{y}_0,\vect{w}_0)} & \leq \psi(\tau)\abs{\vect{w}_0}+\epsilon L_B\left(1-\phi(\tau)\right) + \int_{0}^{\tau}h(\tau-s)\left[\psi(s)|\vect{w}_0|+\epsilon L_B\left(1-\phi(\tau)\right) \right] \ \id s\nonumber\\
          & = \left[\psi(\tau)+\int_{0}^\tau \! h(\tau-s)\psi(s) \ \id s\right]\abs{\vect{w}_0}+\epsilon L_B\left(1-\phi(\tau)\right)\nonumber \\
					&\quad +\epsilon L_B\int_{0}^{\tau}h(\tau-s)\left(1-\phi(s)\right) \ \id s, 
\end{align}
where $h(\tau;\epsilon)=\sum_{j=1}^{\infty}k_j(\tau)$ with $k_1=\epsilon L_M\psi$ and $k_j=k_{j-1}\ast k_1$. Induction on $j$ leads to the expression
$$k_j=(\epsilon L_M)^{j}\psi^{\ast j}.$$
Therefore we have the identity
\begin{align*}
\psi(\tau)+\int_{0}^\tau h(\tau-s)\psi(s) \ \id s & =\frac{k_1(\tau)}{\epsilon L_M}+\int_{0}^{\tau} \!\sum_{j=1}^{\infty}k_j(\tau-s)\frac{k_1(s)}{\epsilon L_M} \ \id s\\
 & = \frac{k_1(\tau)}{\epsilon L_M}+\frac{1}{\epsilon L_M}\sum_{j=1}^{\infty}\int_{0}^{\tau} \! k_j(\tau-s)k_1(s) \ \id s\\
 & = \frac{k_1(\tau)}{\epsilon L_M}+\frac{1}{\epsilon L_M}\sum_{j=1}^{\infty}k_{j+1}(\tau)\\
 & = \frac{1}{\epsilon L_M}\sum_{j=1}^{\infty}k_{j}(\tau)\\
 & = \frac{1}{\epsilon L_M}h(\tau),
\end{align*}
where we omitted the dependence of $h$ on the parameter $\epsilon$ for notational simplicity.

This shows that
\be 
\label{eq:w_ineq1}
\abs{\vect{w}(\tau;\vect{y}_0,\vect{w}_0)} \leq \frac{\abs{\vect{w}_0}}{\epsilon L_M}h(\tau) + \epsilon L_B \left(1-\phi(\tau)\right) + \epsilon L_B \int_{0}^{\tau} h(\tau - s)\left(1-\phi(s)\right) \ \id s.
\ee
Since $0 \leq \phi(\tau) \leq 1$, we have that $(1-\phi(\tau)) \leq 1$ and therefore the inequality can be further simplified to
\be 
\label{eq:w_ineq1_2}
\abs{\vect{w}(\tau;\vect{y}_0,\vect{w}_0)} \leq \frac{\abs{\vect{w}_0}}{\epsilon L_M}h(\tau) + \epsilon L_B\left[1-\phi(\tau) \right]+ \epsilon L_B\int_{0}^{\tau} h(s) \ \id s.
\ee

So far we have assumed that the series $\sum_{j=1}^{\infty}k_j=\sum_{j=1}^{\infty}(\epsilon L_M)^j\psi^{\ast j}$ converges uniformly to a limit $h$. To prove this, we first show that for any $j$ and $\tau\geq 0$, $0\leq \psi^{\ast j}(\tau)\leq 1$. For $j=1$, this property holds since $0\leq \psi\leq 1$. For $j=2$ we have
$$0\leq \psi^{\ast 2}(\tau):=\int_{0}^{\tau}\psi(\tau-s)\psi(s) \ \id s\leq \int_{0}^{\tau}\psi(s) \ \id s=1-\phi(\tau)\leq 1.$$
By induction on $j$, we get $0\leq \psi^{\ast j}(\tau)\leq 1$. As a result, $(\epsilon L_M)^j\psi^{\ast j}\leq (\epsilon L_M)^j$. Since $\epsilon L_M<1$, the series $\sum_{j=1}^{\infty}(\epsilon L_M)^j$ converges. It follows that
\be \label{eq:h_1}
\abs{h(\tau)} \leq \sum_{j=1}^{\infty}(\epsilon L_M)^j = \frac{\epsilon L_M}{1 - \epsilon L_M}
\ee
by summing up the geometric series. By the dominated convergence theorem, the sequence $\sum_{j=1}^{n}(\epsilon L_M)^j\psi^{\ast j}$ converges uniformly to a function $h$ as $n\rightarrow\infty$. Since for any $n$, the series $\sum_{j=1}^{n}(\epsilon L_M)^j \psi^{\ast j}$ is continuous, so is the limiting function $h$. This shows that $h:[0,\infty)\rightarrow \mathbb R$ is continuous and $h\geq 0$.

Now, observe that
\be
\begin{split}
\label{eq:h_2}
\int_{0}^\tau h(\xi) \ \id\xi & =\int_{0}^\tau \sum_{j=1}^{\infty} (\epsilon L_M)^j\psi^{\ast j}(\xi) \ \id \xi \\
       & = \sum_{j=1}^{\infty} (\epsilon L_M)^j \int_{0}^\tau \psi^{\ast j}(\xi) \ \id \xi\\
       & \leq \sum_{j=1}^{\infty} (\epsilon L_M)^j  = \frac{\epsilon L_M}{1 - \epsilon L_M}, \\
\end{split}
\ee
where we used the uniform convergence of the series and the fact that, for any $j$,
\begin{align*}
0\leq \int_{0}^\tau \psi^{\ast j}(\xi)\ \id\xi & \leq\left(\int_{0}^\tau \psi^{\ast (j-1)}(\xi) \ \id\xi\right) \left(\int_{0}^\tau \psi(\xi) \ \id\xi\right) \\
& \leq \cdots \leq \left(\int_{0}^\tau \psi(\xi) \ \id\xi\right)^j= (1-\phi(\tau))^j\leq 1,
\end{align*}
by repeated application of Young's inequality for convolutions. This also shows that $h(\tau)\rightarrow 0$ as $\tau\rightarrow\infty$, since $\abs{h}_1<\infty$ and $h$ is uniformly continuous.

Using inequality \eqref{eq:h_2} in \eqref{eq:w_ineq1_2} and the definition of $h$, we get
\begin{align}
\label{eq:w_ineq2}
\abs{\vect{w}(\tau;\vect{y}_0,\vect{w}_0)} & \leq \frac{\abs{\vect{w}_0}}{\epsilon L_M}h(\tau) + \epsilon L_B\left(1-\phi(\tau)\right) +\frac{\epsilon^2 L_M L_B}{1 - \epsilon L_M} \\
& = \abs{\vect{w}_0}\left[\sum_{j=1}^{\infty}(\epsilon L_M)^{j-1}\psi^{*j}(\tau) \right] + \epsilon L_B\left(1-\phi(\tau)\right) +\frac{\epsilon^2 L_M L_B}{1 - \epsilon L_M}.
\end{align}
This proves part (i) of the theorem.

Taking the sup of $\abs{\vect{w}(\tau;\vect{y}_0,\vect{w}_0)}$ over $[0, \delta)$, we get
\be
\sup_{0\leq \tau<\delta}\abs{\vect{w}(\tau;\vc y_0,\vc w_0)} \leq \frac{\abs{\vect{w}_0} + \epsilon L_B}{1 - \epsilon L_M},
\ee
which proves part (ii) of Theorem \ref{thm:properties}. 

If $\delta = \infty$, then we can take the limitsup of $|\vc w|$. Using inequality \eqref{eq:w_ineq2} we get the asymptotic estimate
\be
\limsup_{\tau\rightarrow\infty}|\vc w(\tau;\vc y_0,\vc w_0)|\leq \frac{\epsilon L_B}{1-\epsilon L_M},
\ee
which proves part (iii) of Theorem \ref{thm:properties}. Here, we used the fact that $\lim_{\tau\rightarrow\infty} h(\tau) = 0$ and $\lim_{\tau\rightarrow\infty} \phi(\tau) = 0$. 
\end{proof}

\section{Proof of Lemma \ref{lemma:WD}}
\label{app:lemma_WD}

Let $\tau_1$, $\tau_2 \in [0, \delta)$. Bound $\abs{\vect{z}_{loc}(\tau_2) - \vect{z}_{loc}(\tau_1)}$ by
\benn
\begin{split}
\abs{\vect{z}_{loc}(\tau_2) - \vect{z}_{loc}(\tau_1)} & \leq  \abs{\vect{y}_{loc}(\tau_2) - \vect{y}_{loc}(\tau_1)} + \abs{\vect{w}_{loc}(\tau_2) - \vect{w}_{loc}(\tau_1)} \nonumber \\
                      & \leq  \abs{\vect{w}_0}\abs{\psi(\tau_2) - \psi(\tau_1)} + \epsilon \int_{\tau_1}^{\tau_2} \! \abs{\vect{w}_{loc}(s)} + \abs{\vect{A_u}(\vect{y}_{loc}(s), s)} \ \id s \nonumber \\
											& + \epsilon \int_{\tau_1}^{\tau_2} \! \psi(\tau_2 - s)\left[\abs{\vect{M_u}(\vect{y}_{loc}(s), s)}\abs{\vect{w}_{loc}(s)} + \abs{\vect{B_u}(\vect{y}_{loc}(s), s)}\right] \ \id s \nonumber \\
											& + \epsilon \int_{0}^{\tau_1} \! \abs{\psi(\tau_2 - s) - \psi(\tau_1 - s)}\left[\abs{\vect{M_u}(\vect{y}_{loc}(s),s)}\abs{\vect{w}_{loc}(s)} + \abs{\vect{B_u}(\vect{y}_{loc}(s), s)}\right] \ \id s \nonumber
\end{split}
\eenn
Without loss of generality suppose $\tau_1 \leq \tau_2$, so that $\abs{\psi(\tau_2 - s) - \psi(\tau_1 - s)} = \psi(\tau_2 - s) - \psi(\tau_1 - s)$. Taking the infinity norm over $[0, \delta)$ to bound $\norm{\vect{M_u}(\vect{y}_{loc}(s),s)}_\infty$, $\norm{\vect{B_u}(\vect{y}_{loc}(s),s)}_\infty$, $\norm{\vect{A_u}(\vect{y}_{loc}(s),s)}_\infty$, and $\abs{\vect{w}_{loc}(s)}$ by Theorem \ref{thm:properties}, we get
\begin{equation*}
\begin{split}
\abs{\vect{z}_{loc}(\tau_2) - \vect{z}_{loc}(\tau_1)} & \leq \abs{\vect{w}_0}\abs{\psi(\tau_2) - \psi(\tau_1)} + \epsilon \left(L_A + \frac{\abs{\vect{w}_0} + \epsilon L_B}{1 - \epsilon L_M}\right) \abs{\tau_2 - \tau_1} \\
                      & + \epsilon \left[\frac{L_M\abs{\vect{w}_0} + L_B}{1 - \epsilon L_M}\right] \left(\abs{\tau_2 - \tau_1} + \int_{0}^{\tau_1} \! \psi(\tau_1 - s) - \psi(\tau_2 - s) \ \id s\right).
\end{split}
\end{equation*}
By the results of Theorem \ref{thm:homogeneous_solution}, $\psi(\tau_1 - s) - \psi(\tau_2 - s) = \phi'(\tau_2 - s) - \phi'(\tau_1 - s) \geq 0$. Integrate and rearrange to obtain
\begin{equation*}
\begin{split}
\abs{\vect{z}_{loc}(\tau_2) - \vect{z}_{loc}(\tau_1)} & \leq \abs{\vect{w}_0}\abs{\psi(\tau_2) - \psi(\tau_1)} + \epsilon \left(L_A + \frac{\abs{\vect{w}_0} + \epsilon L_B}{1 - \epsilon L_M}\right) \abs{\tau_2 - \tau_1} \\
                      & + \epsilon \left[\frac{L_M\abs{\vect{w}_0} + L_B}{1 - \epsilon L_M}\right] \left(\abs{\tau_2 - \tau_1} + \phi(\tau_2) - \phi(\tau_1) \right) \\
                      & + \epsilon \left[\frac{L_1\abs{\vect{w}_0} + L_2}{1 - \epsilon L_1}\right] \left(\phi(0) - \phi(\tau_2 - \tau_1)\right).
\end{split}
\end{equation*}
Since both $\psi$ and $\phi$ are uniformly continuous over $[0, \infty)$ by Theorem \ref{thm:homogeneous_solution}, each of $\abs{\psi(\tau_2) - \psi(\tau_1)}$, $\abs{\phi(\tau_2) - \phi(\tau_1)}$, and $\abs{\phi(0) - \phi(\tau_2 - \tau_1)} \rightarrow 0$ as $\abs{\tau_2 - \tau_1} \rightarrow 0$. Hence $\abs{\vect{z}_{loc}(\tau_2) - \vect{z}_{loc}(\tau_1)} \rightarrow 0$ as $\tau_1$, $\tau_2 \rightarrow \delta_{-}$. 

Now, if we take a sequence $\left. \{t_n\} \right.$ $t_n \in [0, \delta)$ such that $\lim_{n \rightarrow \infty} t_n \rightarrow \delta$, then it follows that $\left. \{\vect{z}_{loc}(t_n)\} \right.$ is a Cauchy sequence. The sequence is convergent in $\mathbb{R}^{2n}$ since $\mathbb{R}^{2n}$ is a complete metric space. The limit is given by the integral equation \eqref{eq:MR_relaxation_nonhomogeneous} evaluated at $\tau = \delta$:
\[\vect{z}_{loc}(\delta) = \left(
  \begin{array}{lr}
    \vect{y}_0 + \epsilon \int_{0}^{\delta} \! \vect{w}_{loc}(s) + \vect{A_u}(\vect{y}_{loc}(s), s) \ \id s \nonumber \\
    \psi(\delta)\vect{w}_0 + \epsilon  \int_{0}^{\delta} \! \psi(\tau - s)\left[-\vect{M_u}(\vect{y}_{loc}(s), s)\vect{w}_{loc}(s) + \vect{B_u}(\vect{y}_{loc}(s), s)\right] \ \id s \nonumber
  \end{array}
\right).
\]
This ends the proof.

\section{Proof of Lemma \ref{lem:continuous}}
\label{app:lemma_continuous}

Let $\vect{\Phi} = \left(\bxi,\betta\right) \in X_K^{\delta,h}$, and $\tau_1$, $\tau_2 \in [\delta, \delta + h)$. Bound $\abs{(\vect{F}\vect{\Phi})(\tau_2) - (\vect{F}\vect{\Phi})(\tau_1)}$ by
\benn
\begin{split}
\abs{(\vect{F}\vect{\Phi})(\tau_2) - (\vect{F}\vect{\Phi})(\tau_1)} & \leq\abs{\vect{\Phi_0}(\tau_2) - \vect{\Phi_0}(\tau_1)} + \epsilon\int_{\tau_1}^{\tau_2} \! \abs{\betta(s)} + \abs{\vect{A_u}(\bxi(s), s)} \ \id s \\
                      & +  \epsilon\int_{\tau_1}^{\tau_2} \! \psi(\tau_2 - s)\left[\abs{\vect{M_u}(\bxi(s), s)}\abs{\betta(s)} + \abs{\vect{B_u}(\bxi(s), s)}\right] \ \id s \\ 
											& + \epsilon \int_{\delta}^{\tau_1} \! \left(\psi(\tau_2 - s) - \psi(\tau_1 - s)\right)\left[\abs{\vect{M_u}(\bxi(s), s)}\abs{\betta(s)} + \abs{\vect{B_u}(\bxi(s), s)}\right] \ \id s,
\end{split}
\eenn
where
\benn
\begin{split}
\abs{\vect{\Phi_0}(\tau_2) - \vect{\Phi_0}(\tau_1)} & \leq \abs{\vect{w}_0}\abs{\psi(\tau_2) - \psi(\tau_1)} \\
                                                    & + \epsilon \int_{0}^{\delta} \! \abs{\psi(\tau_2 - s) - \psi(\tau_1 - s)}\left[\abs{\vect{M_u}(\vect{y}_{loc}(s), s)}\abs{\vect{w}_{loc}(s)} + \abs{\vect{B_u}(\vect{y}_{loc}(s), s)} \right] \ \id s.
\end{split}
\eenn
Without loss of generality suppose $\tau_1 \leq \tau_2$, so that $\abs{\psi(\tau_2 - s) - \psi(\tau_1 - s)} = \psi(\tau_2 - s) - \psi(\tau_1 - s)$. Taking the infinity norm over $[\delta, \delta + h)$ to bound $\norm{\vect{M_u}(\bxi(s), s)}_\infty$, $\norm{\vect{B_u}(\bxi(s), s)}_\infty$, $\norm{\vect{A_u}(\bxi(s), s)}_\infty$, $\norm{\betta(s)}_\infty$, and $\abs{\vect{w}_{loc}(s)}$ by inequality \eqref{eq:int_ineq}, we get
\begin{equation*}
\begin{split}
\abs{(\vect{F}\vect{\Phi})(\tau_2) - (\vect{F}\vect{\Phi})(\tau_1)} & \leq  \abs{\vect{w}_0}\abs{\psi(\tau_2) - \psi(\tau_1)} \\
                      & +     \epsilon \left(\frac{L_M\abs{w_0} + L_B}{1 - \epsilon L_M}\right) \int_{0}^{\delta} \! \psi(\tau_1 - s) - \psi(\tau_2 - s) \ \id s \\
                      & +     \epsilon(K + L_A)\abs{\tau_2 - \tau_1} + \epsilon\left(L_MK + L_B\right)\abs{\tau_2 - \tau_1} \\
											& +     \epsilon\left(L_MK + L_B\right)\int_{\delta}^{\tau_1} \! \psi(\tau_1 - s) - \psi(\tau_2 - s) \ \id s. 
\end{split}
\end{equation*}
By the results of Theorem \ref{thm:homogeneous_solution}, $\psi(\tau_1 - s) - \psi(\tau_2 - s) = \phi'(\tau_2 - s) - \phi'(\tau_1 - s) \geq 0$. Finally, integrate and rearrange to obtain
\begin{equation*}
\begin{split}
\abs{(\vect{F}\vect{\Phi})(\tau_2) - (\vect{F}\vect{\Phi})(\tau_1)} &\leq \abs{\vect{w}_0}\abs{\psi(\tau_2) - \psi(\tau_1)} \\
                      & +     \epsilon \left(\frac{L_M\abs{w_0} + L_B}{1 - \epsilon L_M}\right) \left[\left(\phi(\tau_1 - \delta) - \phi(\tau_2 - \delta)\right) + \left(\phi(\tau_2) - \phi(\tau_1)\right)\right] \\
                      & +     \epsilon(K + L_A)\abs{\tau_2 - \tau_1} + \epsilon\left(L_MK + L_B\right) \abs{\tau_2 - \tau_1}  \\
											& +     \epsilon\left(L_MK + L_B\right)\left[\left(\phi(0) - \phi(\tau_2 - \tau_1)\right) + \left(\phi(\tau_2 - \delta) - \phi(\tau_1 - \delta)\right)\right].
\end{split}
\end{equation*}
Since both $\psi$ and $\phi$ are uniformly continuous over $[0, \infty)$ by Theorem \ref{thm:homogeneous_solution}, each of $\abs{\psi(\tau_2) - \psi(\tau_1)}$, $\abs{\phi(\tau_2) - \phi(\tau_1)}$, $\abs{\phi(\tau_1 -\delta) - \phi(\tau_2 - \delta)}$, and $\abs{\phi(0) - \phi(\tau_2 - \tau_1)} \rightarrow 0$ as $\abs{\tau_2 - \tau_1} \rightarrow 0$. Hence $\abs{(\vect{F}\vect{\Phi})(\tau_2) - (\vect{F}\vect{\Phi})(\tau_1)} \rightarrow 0$ as $\abs{\tau_2 - \tau_1} \rightarrow 0$. This shows that $\vect{F}$ maps $X_K^{\delta,h}$ to a family of uniformly equicontinuous functions in $C(\left[\delta, \delta + h\right); \mathbb{R}^{2n})$.

\end{appendices}


\begin{thebibliography}{36}
\providecommand{\natexlab}[1]{#1}
\providecommand{\url}[1]{\texttt{#1}}
\expandafter\ifx\csname urlstyle\endcsname\relax
  \providecommand{\doi}[1]{doi: #1}\else
  \providecommand{\doi}{doi: \begingroup \urlstyle{rm}\Url}\fi

\bibitem[Galdi et~al.(2008)Galdi, Rannacher, Robertson, and
  Turek]{galdi2008hemodynamical}
G.~P. Galdi, R.~Rannacher, A.~M. Robertson, and S.~Turek.
\newblock \emph{Hemodynamical flows: {M}odeling, Analysis and Simulation},
  volume~37 of \emph{Oberwolfach Seminars}.
\newblock Springer, 2008.

\bibitem[Cartwright et~al.(2010)Cartwright, Feudel, K{\'a}rolyi, de~Moura,
  Piro, and T{\'e}l]{cartwright2010}
J.~H.~E. Cartwright, U.~Feudel, G.~K{\'a}rolyi, A.~de~Moura, O.~Piro, and
  T.~T{\'e}l.
\newblock Dynamics of finite-size particles in chaotic fluid flows.
\newblock In \emph{Nonlinear Dynamics and Chaos: Advances and Perspectives},
  pages 51--87. Springer, 2010.

\bibitem[Stokes(1851)]{stokes1851}
G.~G. Stokes.
\newblock \emph{On the effect of the internal friction of fluids on the motion
  of pendulums}, volume~9.
\newblock 1851.

\bibitem[Basset(1888)]{basset2}
A.~B. Basset.
\newblock \emph{A treatise on hydrodynamics}.
\newblock Deighton, Bell and Co, Cambridge, 1888.

\bibitem[Boussinesq(1885)]{boussinesq}
J.~V. Boussinesq.
\newblock Sur la r\'esistance qu'oppose un fluide ind\'efini au repos, sans
  pesanteur, au mouvement vari\'e d'une sph\'ere solide qu'il mouille sur toute
  sa surface, quand les vitesses restent bien continues et assez faibles pour
  que leurs carr\'es et produits soient n\'egligeables.
\newblock \emph{Comptes Rendu de l'Academie des Sciences}, 100:\penalty0
  935--937, 1885.

\bibitem[Oseen(1927)]{oseen1927}
C.~W. Oseen.
\newblock \emph{Hydrodynamik}.
\newblock Akademische Verlagsgesellschaft, Leipzig, 1927.

\bibitem[Tchen(1947)]{Tchen}
C.~M. Tchen.
\newblock \emph{Mean value and correlation problems connected with the motion
  of small particles suspended in a turbulent fluid}.
\newblock PhD thesis, TU Delft, 1947.

\bibitem[Corrsin and Lumley(1956)]{corrsin1956}
S.~Corrsin and J.~Lumley.
\newblock On the equation of motion for a particle in turbulent fluid.
\newblock \emph{Appl. Sci. Res.}, 6\penalty0 (2):\penalty0 114--116, 1956.

\bibitem[Maxey and Riley(1983)]{MR}
M.~R. Maxey and J.~J. Riley.
\newblock Equation of motion for a small rigid sphere in a nonuniform flow.
\newblock \emph{Phys. Fluids}, 26:\penalty0 883--889, 1983.

\bibitem[Auton et~al.(1988)Auton, Hunt, and Prud'{H}omme]{auton1988}
T.~R. Auton, J.~C.~R. Hunt, and M.~Prud'{H}omme.
\newblock The force exerted on a body in inviscid unsteady non-uniform
  rotational flow.
\newblock \emph{J. of Fluid Mech.}, 197:\penalty0 241--257, 1988.

\bibitem[Maxey(1987)]{IP_maxey87}
M.~R. Maxey.
\newblock The gravitational settling of aerosol particles in homogeneous
  turbulence and random flow fields.
\newblock \emph{J. Fluid Mech.}, 174\penalty0 (1):\penalty0 441--465, 1987.

\bibitem[Balkovsky et~al.(2001)Balkovsky, Falkovich, and Fouxon]{balkovsky2001}
E.~Balkovsky, G.~Falkovich, and A.~Fouxon.
\newblock Intermittent distribution of inertial particles in turbulent flows.
\newblock \emph{Phys. Rev. Lett.}, 86\penalty0 (13):\penalty0 2790, 2001.

\bibitem[Candelier et~al.(2004)Candelier, Angilella, and Souhar]{IP_candelier}
F.~Candelier, J.~R. Angilella, and M.~Souhar.
\newblock On the effect of the {B}oussinesq--{B}asset force on the radial
  migration of a {S}tokes particle in a vortex.
\newblock \emph{Physics of Fluids}, 16\penalty0 (5):\penalty0 1765--1776, 2004.

\bibitem[Toegel et~al.(2006)Toegel, Luther, and Lohse]{toegel2006}
R.~Toegel, S.~Luther, and D.~Lohse.
\newblock Viscosity destabilizes sonoluminescing bubbles.
\newblock \emph{Phys. Rev. Lett.}, 96\penalty0 (11):\penalty0 114301, 2006.

\bibitem[Garbin et~al.(2009)Garbin, Dollet, Overvelde, Cojoc, Di~Fabrizio, van
  Wijngaarden, Prosperetti, de~Jong, Lohse, and Versluis]{gabrin2009}
V.~Garbin, B.~Dollet, M.~Overvelde, D.~Cojoc, E.~Di~Fabrizio, L.~van
  Wijngaarden, A.~Prosperetti, N.~de~Jong, D.~Lohse, and M.~Versluis.
\newblock History force on coated microbubbles propelled by ultrasound.
\newblock \emph{Phys. Fluids}, 21\penalty0 (9), 2009.

\bibitem[Daitche and T{\'e}l(2011)]{daitche2011memory}
A.~Daitche and T.~T{\'e}l.
\newblock Memory effects are relevant for chaotic advection of inertial
  particles.
\newblock \emph{Phys.l Rev. Lett.}, 107\penalty0 (24):\penalty0 244501, 2011.

\bibitem[Guseva et~al.(2013)Guseva, Feudel, and T{\'e}l]{guseva2013influence}
K.~Guseva, U.~Feudel, and T.~T{\'e}l.
\newblock Influence of the history force on inertial particle advection:
  {G}ravitational effects and horizontal diffusion.
\newblock \emph{Phys. Rev. E}, 88\penalty0 (4):\penalty0 042909, 2013.

\bibitem[Daitche and T\'el(2014)]{Daitche_NJP}
A.~Daitche and T.~T\'el.
\newblock Memory effects in chaotic advection of inertial particles.
\newblock \emph{New J. of Phys.}, 16\penalty0 (7):\penalty0 073008, 2014.

\bibitem[Rubin et~al.(1995)Rubin, Jones, and Maxey]{rubin1995_IP}
J.~Rubin, C.~K. R.~T. Jones, and M.~Maxey.
\newblock Settling and asymptotic motion of aerosol particles in a cellular
  flow field.
\newblock \emph{J. Nonlinear Sci.}, 5\penalty0 (4):\penalty0 337--358, 1995.

\bibitem[Mograbi and Bar-Ziv(2006)]{mograbi2006}
E.~Mograbi and E.~Bar-Ziv.
\newblock On the asymptotic solution of the maxey-riley equation.
\newblock \emph{Phys. Fluids}, 18\penalty0 (5):\penalty0 051704, 2006.

\bibitem[Haller and Sapsis(2008)]{IP_haller08}
G.~Haller and T.~Sapsis.
\newblock Where do inertial particles go in fluid flows?
\newblock \emph{Physica D}, 237\penalty0 (5):\penalty0 573--583, 2008.

\bibitem[Maxey(1993)]{MR_initCond}
M.~R. Maxey.
\newblock The equation of motion for a small rigid sphere in a nonuniform or
  unsteady flow.
\newblock In \emph{Gas-solid flows, 1993}, volume 166, pages 57--62. The
  {A}merican society of mechanical engineers, 1993.

\bibitem[Kobayashi and Coimbra(2005)]{kobayashi}
M.~H. Kobayashi and C.~F.~M. Coimbra.
\newblock On the stability of the {M}axey-{R}iley equation in nonuniform linear
  flows.
\newblock \emph{Phys. Fluids}, 17\penalty0 (11):\penalty0 --113301, 2005.

\bibitem[Farazmand and Haller(2014)]{MR_EUR}
M.~Farazmand and G.~Haller.
\newblock The {M}axey-{R}iley equation: Existence, uniqueness and regularity of
  solutions.
\newblock \emph{J. Nonliner Analysis-B}, 2014.
\newblock In press.

\bibitem[Podlubny(1998)]{podlubny1998fractional}
I.~Podlubny.
\newblock \emph{Fractional differential equations: an introduction to
  fractional derivatives, fractional differential equations, to methods of
  their solution and some of their applications}, volume 198.
\newblock Academic press, 1998.

\bibitem[Gorenflo and Mainardi(1997)]{mainardi_lecNotes}
R.~Gorenflo and F.~Mainardi.
\newblock Fractional calculus.
\newblock \emph{Fractals and fractional calculus in continuum mechanics},
  \penalty0 (378):\penalty0 277, 1997.

\bibitem[Shadden et~al.(2005)Shadden, Lekien, and Marsden]{shadden05}
S.~C. Shadden, F.~Lekien, and J.~E. Marsden.
\newblock Definition and properties of {L}agrangian coherent structures from
  finite-time {L}yapunov exponents in two-dimensional aperiodic flows.
\newblock \emph{Physica D}, 212:\penalty0 271--304, 2005.

\bibitem[Daitche(2013)]{daitche2013advection}
A.~Daitche.
\newblock Advection of inertial particles in the presence of the history force:
  Higher order numerical schemes.
\newblock \emph{J. Comput. Phys.}, 254:\penalty0 93--106, 2013.

\bibitem[Kou et~al.(2012)Kou, Zhou, and Li]{kou2012existence}
C.~Kou, H.~Zhou, and C.~Li.
\newblock Existence and continuation theorems of {R}iemann--{L}iouville type
  fractional differential equations.
\newblock \emph{Int. J. of Bifurcation and Chaos}, 22\penalty0 (04), 2012.

\bibitem[Sapsis and Haller(2010)]{sapsis2010clustering}
T.~Sapsis and G.~Haller.
\newblock Clustering criterion for inertial particles in two-dimensional
  time-periodic and three-dimensional steady flows.
\newblock \emph{Chaos}, 20\penalty0 (1):\penalty0 017515, 2010.

\bibitem[Bateman et~al.(1955)Bateman, Erd{\'e}lyi, Magnus, Oberhettinger, and
  Tricomi]{bateman1955higher}
H.~Bateman, A.~Erd{\'e}lyi, W.~Magnus, F.~Oberhettinger, and F.~G. Tricomi.
\newblock \emph{Higher transcendental functions}, volume~3.
\newblock McGraw-Hill New York, 1955.

\bibitem[Haubold et~al.(2011)Haubold, Mathai, and Saxena]{haubold2011mittag}
H.~J. Haubold, A.~M. Mathai, and R.~K. Saxena.
\newblock Mittag-leffler functions and their applications.
\newblock \emph{Journal of Applied Mathematics}, 2011, 2011.

\bibitem[Abramowitz and Stegun(1972)]{Abramowitz+Stegun}
M.~Abramowitz and I.~A. Stegun.
\newblock \emph{Error Function and Fresnel Integrals. Ch. 7 in Handbook of
  Mathematical Functions with Formulas, Graphs, and Mathematical Tables}.
\newblock Dover, New York, 9th printing edition, 1972.

\bibitem[{\relax DLMF}()]{NIST:DLMF}
{\relax DLMF}.
\newblock {NIST Digital Library of Mathematical Functions}.
\newblock http://dlmf.nist.gov/, Release 1.0.8 of 2014-04-25.
\newblock URL \url{http://dlmf.nist.gov/}.
\newblock Online companion to \cite{Olver:2010:NHMF}.

\bibitem[Chu and Metcalf(1967)]{chu1967gronwall}
S.~C. Chu and F.~T. Metcalf.
\newblock On gronwall?s inequality.
\newblock \emph{Proc. Amer. Math. Soc.}, 18\penalty0 (3):\penalty0 439--440,
  1967.

\bibitem[Olver et~al.(2010)Olver, Lozier, Boisvert, and Clark]{Olver:2010:NHMF}
F.~W.~J. Olver, D.~W. Lozier, R.~F. Boisvert, and C.~W. Clark, editors.
\newblock \emph{{NIST Handbook of Mathematical Functions}}.
\newblock Cambridge University Press, New York, NY, 2010.
\newblock Print companion to \cite{NIST:DLMF}.

\end{thebibliography}

\end{document}